\documentclass[journal]{IEEEtran}
\usepackage[T1]{fontenc}% optional T1 font encoding
\usepackage{datetime}
\usepackage{listings}
\usepackage{graphicx}
\usepackage{float}
\usepackage{subfigure} 
\usepackage{cite}
\usepackage{caption} 
\usepackage{color}
\usepackage{amsthm}

\newtheorem{corollary}{Corollary}
\newtheorem{prop}{Proposition}

\usepackage{amsmath} 
 
\usepackage[cmintegrals]{newtxmath} 
  
\begin{document}
\title{Age of Information for Multicast Transmission with Fixed and Random Deadlines in IoT Systems}
\author{\IEEEauthorblockN{Jie Li, Yong Zhou, \textit{Member, IEEE}, and He Chen, \textit{Member, IEEE}}

\thanks{

J. Li is with the School of Information Science and Technology, ShanghaiTech University, Shanghai 201210, China, also with the Shanghai Institute of Microsystem and Information Technology, Chinese Academy of Sciences, Shanghai 200050, China, and also with the University of Chinese Academy of Sciences, Beijing 100049, China (E-mail: lijie3@shanghaitech.edu.cn)

Y. Zhou is with the  School of Information Science and Technology, ShanghaiTech University, Shanghai, 201210, China (E-mail: zhouyong@shanghaitech.edu.cn). 

H. Chen is with the Department of Information Engineering, The Chinese University of Hong Kong, Hong Kong, China (E-mail: he.chen@ie.cuhk.edu.hk).} 
} 

\maketitle

\begin{abstract}
	In this paper, we consider the multicast transmission of a real-time Internet of Things (IoT) system, where an access point (AP) transmits time-stamped status updates to multiple IoT devices.
	Different from the existing studies that only considered multicast transmission without deadlines, we enforce a deadline for the service time of each multicast status update, taking into account both the fixed and randomly distributed deadlines. 
	In particular,	a status update is dropped when either its deadline expires or it is successfully received by a certain number of IoT devices. 
	Considering deadlines is important for many emerging IoT applications, where the outdated status updates are of no use to IoT devices. 
	We evaluate the timeliness of the status update delivery by applying a recently proposed metric, named the age of information (AoI), which is defined as the time elapsed since the generation of the most recently received status update. 
	After deriving the distributions of the service time for all possible reception outcomes at IoT devices, we manage to obtain the closed-form expressions of both the average AoI and the average peak AoI.
	Simulations validate the performance analysis, which reveals that the multicast transmission with deadlines achieves a lower average AoI than that without deadlines and there exists an optimal value of the deadline that can minimize the average (peak) AoI. 
	Results also show that the fixed and random deadlines have respective advantages in different deadline regimes. 
	\end{abstract}
	
	
	%\newpage 
	
	\section{Introduction}
Internet of Things (IoT), as a worldwide network of interconnected objects, provides ubiquitous wireless connectivity and automated information delivery for a large amount of smart devices that have the capabilities of monitoring, processing, and communication, and hence being able to support a variety of services \cite{gubbi2013internet, kosta2017age}. 
With pervasive connectivity, the timeliness of fresh information delivery to multiple IoT devices is critical for many emerging IoT applications. 
For example, in a smart parking lot, an access point (AP) continuously collects the occupancy information of all parking spaces and reports the locations of the vacant parking spaces to the nearby drivers within a certain deadline. 
For video streaming in a sport stadium, many audiences sitting in the back are interested in watching the same real-time video, which has a hard deadline constraint and is of no use after the deadline \cite{kim2014scheduling}. 
In addition, in connected vehicle networks, the status updates of autonomous vehicles, including the safety messages (e.g., accident, emergency braking, and traffic congestion) and the non-safety messages (e.g., vehicle position, speed, and heading), are required to be timely delivered to the nearby vehicles and roadside units (RSU) \cite{lyu2019characterizing}. 
These messages with diverse importance usually have different deadline requirements, which can be assumed to follow a random distribution, as in \cite{mao2014optimal, akyurek2018optimal, on}. 
In all these examples, the latest status updates (e.g., vacancy information, live video, safety and non-safety messages) are required to be disseminated to multiple receivers within certain deadlines. 
Hence, enhancing information freshness for multicast transmission in IoT networks with deadlines is critical.
%Typical IoT applications, including environment monitoring, smart metering, connected/autonomous transportation, and remote automation, require reliable and timely information delivery, i.e., the status observed at receivers needs to be fresh. 

The conventional performance metrics (e.g., throughput and delay) cannot adequately capture the information freshness. 
In particular, due to random network delay, maximizing the throughput or minimizing the delay does not necessarily guarantee the freshest information to be observed at the receivers, and hence may lead to the wastage of precious spectrum resources \cite{sun2017update}. 
The \textit{age of information} (AoI), as a powerful performance metric, has recently been proposed to characterize the freshness of information from the receiver's point of view \cite{statusqueue}. 
 The AoI at a receiver is defined as the time difference between the current time and the generation time of the most recently received status update.  
 Hence, both the generation time and the latency of status updates can be captured by the AoI. 
On the other hand, the \textit{peak AoI} refers to the maximum value of AoI right before successfully receiving a status update. 
Motivated by the emerging IoT applications, we are interested in studying the average (peak) AoI of multicast transmission with deadlines, which remains unexplored to the best of our knowledge.

%{\color{blue}In order to  quantify the information freshness, \textit {age of information} (AoI)  and \textit {peak age of information} (PAoI) have been recently  introduced. Particularly, AoI measures the time elapsed since  the latest received update packet was generated, while PAoI  provides information about the maximum value of AoI for  each update and captures the extent to which the update information is stale. Unlike many conventional metrics, e.g., delay  or throughput, AoI and PAoI are affected not only  by the transmission delay but also by the update generation rate, and hence they are more essential and comprehensive for  information freshness evaluation.}

\subsection{Related Works}
	The AoI performance has recently been studied in various systems \cite{real, ontheage, yates2012real, optimizing, kuang2019age, boundson, gu2019timely, tandem, markov, multiplesources}. 
	In particular, the authors in \cite{real} developed a theoretical performance analysis framework for the average AoI under various queueing models (i.e., $M/M/1$, $M/D/1$, and $D/M/1$) by using tools from queueing theory and assuming that the status updates are served in a first-come first-serve (FCFS) manner. 
	It has been demonstrated in \cite{real} that minimizing the average AoI is different from minimizing the average delay.
	The analytical framework developed in \cite{real} was then extended to investigate the impact of the buffer size \cite{ontheage} and the server number \cite{yates2012real} on the average AoI, respectively. 
	Results in \cite{ontheage} and \cite{yates2012real} showed that the average AoI can be decreased by reducing the buffer size and/or increasing the server number. 
	By taking into account the heterogeneous distributions of the service time, the authors in \cite{optimizing} derived the average AoI for an $M/G/1$ queueing model. 
	The AoI performance was also analyzed for mobile edge computing (MEC) networks with computation-intensive tasks, where both the local and remote computing strategies were considered \cite{kuang2019age}. 
	Results in \cite{kuang2019age} showed that remote computing outperforms local computing only when the computation capacity of the edge server is far superior than that of the local device. 
	The authors in \cite{boundson} evaluated the freshness of channel state information (CSI) in terms of the AoI, where the lower bounds for the maximum and average staleness of a greedy CSI dissemination scheme were derived. 
	Besides, the tradeoff between AoI and energy efficiency for unicast transmission was characterized in \cite{gu2019timely}, where a limited number of retransmissions were allowed for each status update. 
	Moreover, the average AoI was also analyzed for wireless networks with queues in tandem \cite{tandem}, with Markov channels \cite{markov}, and with multiple sources \cite{multiplesources}. 
	
	Developing optimal scheduling policies for AoI minimization is another important research direction\cite{schedulingstatus, achievingthe, arafa2018age, elgabli2019reinforcement,
	multipleflows, economicsof,  agebased, adaptivecoding, minimizingthe, uav}. 
	The authors in \cite{schedulingstatus} proposed an age-optimal threshold policy to minimize the average AoI achieved by an energy-harvesting sensor, which is restricted by the time-varying energy arrivals and the battery capacity. 
For energy harvesting networks, the optimal scheduling policy for age-energy tradeoff and the online scheduling policy for AoI minimization were proposed in \cite{achievingthe} and \cite{arafa2018age}, respectively. 
	Moreover, the authors in \cite{elgabli2019reinforcement} and \cite{multipleflows} developed reinforcement learning (RL) based algorithms to minimize the average AoI for ultra-reliable low-latency communication (URLLC) and multi-flow networks, respectively. 
	To balance the tradeoff between the AoI and the sampling cost, the authors in \cite{economicsof} proposed two non-monetary trigger-and-punishment mechanisms to achieve social optimal for scenarios with complete and incomplete information, respectively. 
	Besides, a scheduling scheme was proposed in \cite{agebased} to enhance the timely throughput for unicast transmission with deadlines. 
	An adaptive coding scheme was proposed in \cite{adaptivecoding} to enhance the AoI performance of the user with weak channel conditions. % in a two-user broadcast erasure channel. 
	The authors in \cite{minimizingthe} minimized the AoI for networks with stochastic arrivals under any queue discipline. 
	The peak AoI minimization problem was also studied in unmanned aerial vehicular (UAV) networks \cite{uav}. 
	It is worth noting that all the aforementioned studies focused on the status update systems with unicast transmission.

	Multicast transmission is a spectrum and energy efficient information delivery scheme and can simultaneously serve multiple devices that are interested in the same information.	
	The research on evaluating and optimizing the AoI of multicast transmission has recently received increasing attention \cite{multicast, status, optimum,Buyukates2018Age, broadcastage, schedulingbroadcast}. 
	 The authors in \cite{multicast} and \cite{status} derived the average AoI of a multicast system, where a status update is dropped if it has been successfully received by enough number of receivers. 
	The tradeoff between energy efficiency and average AoI in multicast systems was studied in \cite{optimum}, where a scheduling strategy based on the optimum stopping theory was proposed.
	In \cite{Buyukates2018Age}, the authors studied the average AoI in a two-hop multicast network. 
	The authors in \cite{broadcastage} analyzed the average AoI for broadcast transmission, in which the instantaneous AoI is reduced only when all receivers have received a status update. 
	In addition, the authors in \cite{schedulingbroadcast} proposed several scheduling policies to minimize the average AoI for broadcast transmission over unreliable channels.
	However, the aforementioned studies on multicast transmission did not take into account the deadline.
	This is crucial for many real-time multicast applications, where the status updates are useless to the receivers after the deadline expires.
	It has been demonstrated in \cite{on} and \cite{analysis} that the packet deadline has a significant impact on the average AoI of unicast transmission.
	Specifically, the authors in \cite{on} and \cite{analysis} derived the closed-form expressions of the average AoI for $M/M/1$ and $M/G/1$ queueing systems, respectively, where the waiting time of each packet is subject to a deadline but the service time can be arbitrary large. 

%	as the stale information have little value to be delivered

	\subsection{Main Contributions}
	In this paper, we consider a real-time status update system, where an AP transmits time-sensitive multicast information to multiple IoT devices.
%	{\color{blue} In particular, the deadline-driven real-time mobile IoT applications (e.g., cognitive assistance and object recognition), computing tasks are usually multicasted to the fog node for processing with deadlines \cite{yao2018qos}. }
	Different from the existing studies that considered either multicast transmission without deadlines \cite{status} or unicast transmission with deadlines for the waiting time \cite{on}, we enforce a deadline for the service time of each status update in multicast transmission. 
	We take into account both the \textit{fixed and randomly distributed deadlines} to fully understand the impact of deadlines on the AoI performance. 
	Each status update is time-stamped and transmitted by the AP once it is generated. 
	The multicast transmission of a status update is terminated as soon as its deadline expires or it is successfully received by a sufficient number of devices. 
	%	It is worth noting that the instantaneous AoI evolution is more complicated for multicast transmission with deadlines than that for the existing studies \cite{status,on}.
	%	This is because, the instantaneous AoI evolution in this paper depends on both the reception outcomes of multiple IoT devices and the deadline, whereas that in existing studies only depends on one of these important factors.
	The evolution of the instantaneous AoI for multicast transmission in IoT networks with deadlines is more complicated than that of networks considering either unicast transmission \cite{status} or deadline \cite{on}, making the analysis of the average AoI more challenging. 
	In particular, the instantaneous AoI evolution in this paper depends on both the reception outcomes of multiple IoT devices and the deadline, both of which can be random and should be taken into account when analyzing the average AoI. 
	In contrast, the instantaneous AoI evolution of the existing studies only depends on either the reception outcomes of multiple devices or the deadline.
	We explicitly show that the AoI evolution of multicast transmission with deadlines depends on various parameters, including the service time of multiple devices, deadline, and number of devices required to successfully receive each status update. 
	The main contributions of this paper are summarized as follows. 
	
	\begin{itemize}
	\item We derive the probability density functions (PDFs) of the service time by using order statistics for all possible reception outcomes at the receiving IoT devices, and calculate the first and second moments of the inter-generation time of two consecutive status updates. 
	
	\item We derive the closed-form expressions of both the average  AoI and the average peak AoI for multicast transmission with fixed and randomly distributed deadlines. The analytical results are general and can be easily extended for multicast transmission without deadlines, broadcast transmission with deadlines, and unicast transmission with deadlines.
	The theoretical analysis can be used to quickly evaluate the information freshness at each IoT device for given network parameters and provide a useful guidance on the network parameter setting for enhancing the information freshness.
	
	\item Simulation results validate the theoretical performance analysis and unveil the impact of various parameters on the average (peak) AoI. 
	Results also reveal that the average (peak) AoI of multicast transmission with deadlines is lower than that without deadlines, and the deadline can be further optimized to reduce the average (peak) AoI. 
%	The average (peak) AoI performance of the fixed deadline is similar to that of the random deadline. 
	The fixed and random deadlines have respective advantages in the low and high deadline regimes. 
	Moreover, the fixed deadline is able to achieve a lower minimum average (peak) AoI than the random deadline when optimizing the deadline. 

	\end{itemize}
	
	The rest of this paper is organized as follows. In Section \ref{model}, we describe the system model and the AoI evolution. The average (peak) AoI of multicast transmission with fixed and randomly distributed deadlines are analyzed in Section \ref{fixed} and Section \ref{random}, respectively. 
	The numerical results are presented in Section \ref{sim}. Finally, Section \ref{con} concludes this paper.
	
		\begin{figure}
		\centering
		\includegraphics[scale=0.92]{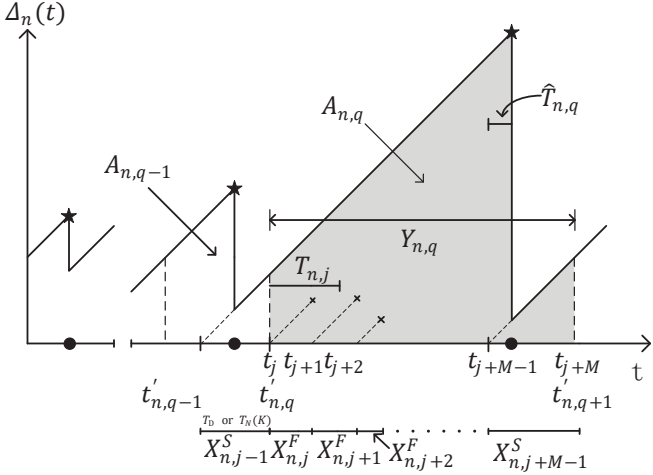}
		\caption{Age evolution of device $n$ over time with deadlines. The time instances that device $n$ successfully receives status updates are marked by $\bullet$, while the time instances immediately before device $n$ successfully receiving status updates are marked by $\star$.}
		\label{fig:aoi}
		 \vspace{-0mm}
	\end{figure}
	
	\section{System Model} \label{model}
	Consider a real-time status update IoT system, where a single AP transmits multicast information with deadlines to multiple IoT devices. 
	We denote $\mathcal{J} = \{1, 2, \ldots, j, \ldots\}$ and $\mathcal{N} = \{1, 2, \ldots, n, \ldots,N \}$ as the index sets of status updates and receiving devices, respectively. 
	We assume that all status updates have the same length in bits. 
	Once a status update is generated, it is time-stamped and transmitted by the AP. 
	The time required to successfully deliver status update $j$ from the AP to device $n$ is denoted as $T_{n,j}$. 
	To account for random channel fading, we assume that $\{T_{n,j}, n \in \mathcal{N}, j \in \mathcal{J}\}$ are independent and exponentially distributed with rate $\lambda_s$ and positive constant shift $c$, as in \cite{multicast, status}. 
	Note that the positive constant shift is considered to account for the same length of status updates and mitigate the probability that a status update can be delivered in an extremely short time.
	Hence, the cumulative distribution function (CDF) of $T_{n,j}$ can be expressed as $F_{T_{n,j}}(t) = 1 - \mathrm{e}^{- \lambda_s (t-c)}, t > c$. 
	A status update is considered to be \textit{served} when it is successfully received by at least $K$ devices for multicast transmission, where $K \le N$, as in \cite{multicast} and \cite{status}. 
	After successfully receiving a status update, a device sends an acknowledgment (ACK) packet back to the AP via an error-free and delay-free control channel. 
	We consider that status update $j \in \mathcal{J}$ is subject to a deadline, denoted as $T_{\mathrm{D},j}$.
	If a status update is not served (i.e., less than $K$ devices successfully receive the status update) when the deadline expires, then this status update is considered useless for the devices that have not successfully received it. 
	As a result, the AP stops transmitting and \textit{drops} this status update.
	The AP \textit{terminates}  the transmission of the current status update (e.g., $j$) if it is either served or dropped. 
	%	Subsequently, the AP generates and transmits a new time-stamped status update (e.g., $j+1$). 
	As soon as the transmission of the current status update (e.g., $j$) is terminated, the AP generates a new time-stamped status update (e.g., $j + 1$).
	
	By denoting $u_n(t)$ as the generation time of the most recently received status update at device $n$ as of time $t$, the instantaneous AoI of device $n$ at time $t$ can be expressed as $\Delta_n(t) = t - u_n(t)$. 
	We depict the evolution of the instantaneous AoI at device $n$ over time as a sawtooth pattern, as shown in Fig. \ref{fig:aoi}.
	As can be observed, the instantaneous AoI increases linearly with time $t$ and drops to a smaller value until a new status update containing fresher information is received.

	To better describe the AoI evolution, we first present the following definitions.
	We denote $t_j$ as the time instant that the AP generates status update $j \in \mathcal{J}$.  
	We define $X_{n,j}^{\mathrm{F}} = t_{j+1} - t_j $ as the inter-generation time of two consecutive status updates $j$ and $j+1$ if status update $j$ is not successfully received by device $n$.
	Similarly, we define $X_{n,j}^{\mathrm{S}} = t_{j+1} - t_j$ as the inter-generation time of two consecutive status updates $j$ and $j+1$ if status update $j$ is successfully received by device $n$. 
	Due to the randomness of service time $T_{n,j}$ and the limitation of the deadline, it is possible that some status updates cannot be successfully received by device $n$.
	Hence, we further denote $t_{n,q}'$  as the termination time of a status update, which corresponds to the $(q-1)$-th status update that has been successfully received by device $n$.
	As shown in {Fig. \ref{fig:aoi}}, $t_{n,q}'=t_j$ implies that status update $(j-1)$ transmitted by the AP is the $(q-1)$-th status update successfully received by device $n$, where $j \ge q$.
	Note that we use subscripts $j$ and $q$ to index the status updates transmitted by the AP and successfully received by the IoT device, respectively. 
	
	As $\{T_{n,j}, n \in \mathcal{N}, j \in \mathcal{J}\}$ are independent and identically distributed (i.i.d.), the evolution processes of the instantaneous AoI  for all devices are statistically identical and hence each device ends up having the same average AoI, which allows us to focus on analyzing the average AoI of device $n$, denoted as $\bar\Delta_n$, for the rest of the paper.
	We denote $\mathcal{Q(\mathcal{T})}=\max\{q|t_{n,q}'\le \mathcal{T}\}$ as the number of status updates that have been received by device $n$ by time $\mathcal{T}$.
	As in \cite{real}, the average AoI of device $n$ can be calculated by
	\begin{equation}
		 \begin{split}
	\bar{\Delta}_n&=\lim\limits_{\mathcal{T} \to \infty }{\frac{1}{\mathcal{T}}\int_{0}^{\mathcal{T}} \Delta_n(t)} \, dt \\
	&=\lim\limits_{\mathcal{T} \to \infty }\frac{Q(\mathcal{T})}{\mathcal{T}}\frac{1}{Q(\mathcal{T})}\sum_{q=1}^{Q(\mathcal{T})} A_{n,q} \\
	& =\frac{\mathbb{E}[A_{n,q}]}{\mathbb{E}[Y_{n,q}]},
	\label{ex3}
	\end{split}
	\end{equation}
	where $\frac{Q(\mathcal{T})}{\mathcal{T}}$ is the steady-state rate of the update delivery, $A_{n,q}$ is the area of the shaded polygon under the sawtooth curve in {Fig. \ref{fig:aoi}}, and $Y_{n,q} = t_{n,q+1}' - t_{n,q}'$ denotes the time duration from the termination time of the $(q-1)$-th status update to that of the $q$-th status update at device $n$.
	Based on Fig. \ref{fig:aoi}, we found the area of the shaded polygon, i.e., $A_{n,q}$, can be expressed as 
	\begin{equation}
		\begin{split}
		A_{n,q} = \, & (X_{n,j-1}^\mathrm{S} + W_{n,q})\hat{T}_{n,q}+{(X_{n,j-1}^\mathrm{S}+\frac{1}{2}W_{n,q})W_{n,q}} \\
		& + \frac{1}{2} {\left({X_{n,j + M_{n,q} - 1}^\mathrm{S}}\right)^2},
		\end{split}
	\label{ex4}
	\end{equation}
where $X_{n,j-1}^\mathrm{S}$ is the inter-generation time of status updates $j-1$ and $j$ when status update $j-1$ is successfully received by device $n$, $\hat{T}_{n,q}$ denotes the service time of the $q$-th status update successfully delivered to device $n$, $M_{n,q}$ is the number of status updates transmitted by the AP within $\left[t_{n,q}', t_{n,q+1}'\right)$, and $W_{n,q}=\sum_{i=j}^{j+M_{n,q}-2}X_{n,i}^\mathrm{F}$ is the summation of $M_{n,q} - 1$ continuous inter-generation times within which all status updates are failed to be received by device $n$.
As $\{X_{n,j}^\mathrm{F}, j \in \mathcal{J}\}$ are i.i.d., we denote $\mathbb{E}[X_{n,j}^\mathrm{F}] = \mathbb{E}[X_{n}^\mathrm{F}]$.
	As $X_{n,j-1}^\mathrm{S}$, $\hat{T}_{n,q}$, and $X_{n,j}^\mathrm{F}$ are independent of each other, the expectation of $A_{n,q}$ can be expressed as
	\begin{equation}
		 \begin{split}
	  \mathbb{E}[A_{n,q}] & =\mathbb{E}[X_{n,j-1}^\mathrm{S}]\mathbb{E}[\hat{T}_{n,q}] + \mathbb{E}[W] \mathbb{E}[\hat{T}_{n,q}] \\
		& \hspace{3mm} + \mathbb{E}[X_{n,j-1}^\mathrm{S}]\mathbb{E}[W]+\frac{1}{2}\mathbb{E}\left[W^2\right] \\ 
		& \hspace{3mm} +\frac{1}{2}\mathbb{E}\left[ \left({X_{n,j+M-1}^\mathrm{S}}\right)^2\right],
	\label{ex6}
	\end{split}
	\end{equation}
	where $X_{n,j-1}^\mathrm{S}$ and $X_{n,j+M-1}^\mathrm{S}$ are identically distributed.
	Hence, we have $\mathbb{E} \left[ X_{n,j-1}^\mathrm{S} \right] = \mathbb{E} \left[ X_{n,j+M-1}^\mathrm{S} \right]$, which is further denoted by $\mathbb{E} \left[ X_{n}^\mathrm{S} \right]$. And $\mathbb{E}[W_{n,q}]=\mathbb{E}[W]$, $\mathbb{E}[M_{n,q}]=\mathbb{E}[M]$.
	As a result, we rewrite \eqref{ex6} as
	\begin{equation}
		 \begin{split}
	  \mathbb{E}[A_{n,q}]= &\frac{1}{2} \mathbb{E}\left[W^2\right] + \left( \mathbb{E}\left[\hat{T}_{n,q}\right] + \mathbb{E}\left[X_{n}^\mathrm{S}\right] \right) \mathbb{E}[W]  \\
		& +\mathbb{E}\left[X_{n}^\mathrm{S}\right]\mathbb{E}\left[\hat{T}_{n,q}\right] + \frac{1}{2}{\mathbb{E}\left[{\left(X_{n}^\mathrm{S}\right)}^2\right]}.
	\label{ex5}
	\end{split}
	\end{equation}
	
	On the other hand, the time duration of the shaded polygon is $Y_{n,q}=W+X_{n,j+M-1}^\mathrm{S}$, the expectation of which is given by
	\begin{equation}
	  \mathbb{E}[Y_{n,q}]=\mathbb{E}[W]+\mathbb{E}\left[X_{n}^\mathrm{S}\right].
	\label{ex7}
	\end{equation}
	
	By substituting (\ref{ex5}) and (\ref{ex7}) into (\ref{ex3}), we have 
	\begin{equation}
	 \begin{split}
		 \bar{\Delta}_n =   \frac{\mathbb{E}[A_{n,q}]}{\mathbb{E}[Y_{n,q}]} = &\frac{\mathbb{E}\left[W^2\right]  + 2 \left( \mathbb{E}\left[\hat{T}_{n,q}\right] + \mathbb{E}\left[X_{n}^\mathrm{S}\right] \right) \mathbb{E}[W]  }{2\mathbb{E}[W]+2\mathbb{E}\left[X_{n}^\mathrm{S}\right]}  \\
		  &+\frac{2 \mathbb{E}\left[X_{n}^\mathrm{S}\right]\mathbb{E}\left[\hat{T}_{n,q}\right] + \mathbb{E}\left[{\left(X_{n}^\mathrm{S}\right)}^2\right]  }{2\mathbb{E}[W]+2\mathbb{E}\left[X_{n}^\mathrm{S}\right]}.
		 \label{AveAoI}
	 \end{split}
	\end{equation}
	
	Average peak AoI is another important performance metric that is closely related to the average AoI and characterizes the worse case AoI. 
	In particular, the $q$-th peak AoI of device $n$ is defined as the value of the instantaneous AoI immediately before it successfully receives the $q$-th status update. 
	Taking the sample path plotted in Fig. \ref{fig:aoi} as an example, the time instances corresponding to the peak AoI of device $n$ are marked by $\star$. 
	Mathematically, the average peak AoI of device $n$ can be calculated by
	\begin{equation}
		\begin{split}
			\bar{P}_{n} =   \mathbb{E}[X_n^S] + \mathbb{E}[W] + \mathbb{E}[\hat{T}_{n,q}].
			\label{PeakAoI}
		\end{split}
	\end{equation}

	To obtain the closed-form expressions of $\bar{\Delta}_n$ and $\bar{P}_n$, we need to calculate all the expectation terms in (\ref{AveAoI}) and (\ref{PeakAoI}). 
	It is worth noting that all the expectation terms in (\ref{AveAoI}) and (\ref{PeakAoI}) depend on the deadline associated with the status updates, as will be demonstrated in the following two sections. 
	To fully illustrate the impact of the deadlines on the AoI, we consider two categories of deadlines, i.e., fixed deadline and randomly distributed deadline. 
	In particular, we shall derive the average (peak) AoI for the cases with fixed and randomly distributed deadlines in Sections \ref{fixed} and \ref{random}, respectively. 

\section{Analysis of Average (Peak) AoI with Fixed Deadlines} \label{fixed}
In this section, we analyze of the average (peak) AoI of multicast transmission with fixed deadlines by deriving the closed-form expressions of all the expectations in (\ref{AveAoI}). 
As a fixed deadline for each status update is considered in this section, we denote $T_{\mathrm{D}} = T_{\mathrm{D},j}, \forall \, j \in \mathcal{J}$, for ease of notations.

\subsection{First and Second Moments  of Inter-Generation Time $X_n^{\mathrm{F}}$  for Fixed Deadline Case} \label{SubSec_XFN}
We first calculate the expectation of the inter-generation time of two consecutive status updates when the former status update is not successfully received by device $n$, i.e., $\mathbb{E}[X_{n}^\mathrm{F}]$. 
Recall that the AP terminates the transmission of a status update when one of the following two events occurs: 1) Event I - The deadline of the status update expires; 2) Event II - At least $K$ devices successfully receive the status update ahead of device $n$. Thus, device $n$ fails to receive the status update if $T_{n,j} > \min\{T_{\mathrm{D}}, T_{N}(K) \}$, where $T_{N}(K)$ is defined as the time duration that $K$ devices have successfully received the status update and it is the $K$-th smallest variable in set $\{T_{n,j}, n \in \mathcal{N}\}$. 
Based on order statistics \cite{order}, the PDF of $T_{N}(K)$ is given by 
\begin{eqnarray} \label{eqn:tnk9}
f_{T_{N}(K)}(t) = K\binom{N}{K} \left(F_{T_{n,j}}(t)\right)^{K-1} \left(1-F_{T_{n,j}}(t) \right)^{N-K} \!\! f_{T_{n,j}}(t),
\end{eqnarray}
where $F_{T_{n,j}}(t) = 1 - \mathrm{e}^{- \lambda_s (t-c)}, t > c$.

We denote the case that device $n$ fails to receive the status update as $\mathcal{C}_{\mathrm{F}}$.
When $T_{n,j} > \min\{T_{\mathrm{D}}, T_{N}(K) \}$, due to the randomness of service times, $X_n^{\mathrm{F}}$  behaves differently for the following two cases: (1) $\mathcal{C}_{\mathrm{F},1}$ - Event II occurs earlier than Event I (i.e., $T_{N}(K) < T_{\mathrm{D}}$); (2) $\mathcal{C}_{\mathrm{F},2}$ - Event I occurs earlier than Event II (i.e., $T_{\mathrm{D}} < T_{N}(K)$).
When Case $\mathcal{C}_{\mathrm{F},1}$ occurs, the instantaneous AoI of device $n$ increases by $T_{N}(K)$ (i.e., $X_{n}^\mathrm{F} = T_{N}(K)$).
On the other hand, when Case $\mathcal{C}_{\mathrm{F},2}$ occurs, the instantaneous AoI of device $n$ increases by $T_{\mathrm{D}}$ (i.e., $X_{n}^\mathrm{F} = T_{\mathrm{D}}$). 
Hence, the expectation of inter-generation time $X_{n}^{\mathrm{F}}$ is given by
\begin{eqnarray}
\mathbb{E}\left[ X_{n}^{\mathrm{F}} \right] = \mathbb{P} \left( \mathcal{C}_{\mathrm{F},1}  \right) \mathbb{E}\left[T_N(K) \left| \mathcal{C}_{\mathrm{F},1}\!\right.\right] + \mathbb{P} \left( \mathcal{C}_{\mathrm{F},2}  \right) T_{\mathrm{D}},
\label{eqn_xnf7}
\end{eqnarray}
where $\mathrm{P}(\mathcal{C}_{\mathrm{F},1})$ and $\mathrm{P}(\mathcal{C}_{\mathrm{F},2})$ denote the probabilities that Cases $\mathcal{C}_{\mathrm{F},1}$ and $\mathcal{C}_{\mathrm{F},2}$ occur when device $n$ fails to receive the status update, respectively, with $\mathrm{P}(\mathcal{C}_{\mathrm{F},1}) + \mathrm{P}(\mathcal{C}_{\mathrm{F},2}) = 1$.
%where $\mathbb{E}\left[T_N(K) \left| \mathcal{C}_{\mathrm{F},1}\!\right.\right]$ is given in (\ref{eqn8}).
Similarly, the second moment of inter-generation time $X_{n}^{\mathrm{F}}$ can be expressed as 
\begin{eqnarray}
\mathbb{E} \left[ \left(X_n^{\mathrm{F}} \right)^2 \right] = \mathbb{P} \left( \mathcal{C}_{\mathrm{F},1} \right) \mathbb{E}[T_N^2(K)| \mathcal{C}_{\mathrm{F},1} ] + \mathbb{P} \left( \mathcal{C}_{\mathrm{F},2}  \right) T_{\mathrm{D}}^2.
 \label{ex17}
\end{eqnarray}
%where $\mathbb{E}\left[T_N^2(K)| \mathcal{C}_{\mathrm{F},1} \right]$ is given in (\ref{eqn9}).
%We calculate $\mathbb{P} \left( \mathcal{C}_{\mathrm{F},1} \right)$,
%$\mathbb{E}\left[T_N(K) \left| \mathcal{C}_{\mathrm{F},1}\!\right.\right] $,
% $ \mathbb{P} \left( \mathcal{C}_{\mathrm{F},2} \right)$,
% and $\mathbb{E}[T_N^2(K)| \mathcal{C}_{\mathrm{F},1} ]$ as follows.

To calculate (\ref{eqn_xnf7}) and (\ref{ex17}), we first derive the first and second moments of conditional $T_N(K)$, i.e., $\mathbb{E}\left[T_N(K) \left| \mathcal{C}_{\mathrm{F},1} \right. \right]$ and $\mathbb{E}\left[T_N^2(K) \left| \mathcal{C}_{\mathrm{F},1} \right. \right]$, in the following proposition. 

\begin{prop}
The first and second moments of the time duration that $K$ devices successfully receive a status update (i.e., $T_{N}(K)$) conditioning on the occurrence of Case $\mathcal{C}_{\mathrm{F},1}$ are
\begin{eqnarray}
\label{eqn8}   
\label{Eq_TNK1} \hspace{-6mm} \mathbb{E}\left[T_N(K) \left| \mathcal{C}_{\mathrm{F},1} \right. \right] \hspace{-6mm} && = \frac{1}{1 - \mathcal{Z}_{K}} \sum_{j=0}^{K-1} \frac{B_{K,j} }{\lambda_s U_{K,j}^2}\Big[1+c\lambda_s U_{K,j}  \nonumber \\
\hspace{-6mm} && \hspace{3mm} - \; (1 + T_{\mathrm{D}}\lambda_s U_{K,j}  )V_{K,j}\Big], \\
\label{Eq_TNK2} \hspace{-6mm}  \mathbb{E}\left[T_N^2(K) \left| \mathcal{C}_{\mathrm{F},1} \right. \right]  \label{eqn9} \hspace{-6mm} && = \frac{1}{1-\mathcal{Z}_{K}} \sum_{j=0}^{K-1} \frac{B_{K,j}}{\lambda_s^2 U_{K,j}^3} \Big[ (1 + c \lambda_s U_{K,j})^2 \nonumber \\
\hspace{-6mm} && \hspace{3mm}  + 1  - \left((1+T_{\mathrm{D}}\lambda_s U_{K,j}  )^2 +1\right)V_{K,j} \Big], 
% \label{eqn9}  \mathbb{E}\left[T_N^2(K) \left| \mathcal{C}_{\mathrm{F},1} \right. \right] \hspace{-6mm} &&= B_{K,j} \left[ \frac{(1 + c \lambda_s U_{K,j})^2+1 }{\lambda_s^2 U_{K,j}^3\left(1-\mathcal{Z}_{K,j}\right)}\right. \nonumber \\
% \hspace{-6mm} && \hspace{3mm} \left. - \frac{\left((1+T_{\mathrm{D}}\lambda_s U_{K,j}  )^2 +1\right)V}{\lambda_s^2 U_{K,j}^3\left(1-\mathcal{Z}_{K,j}\right)} \right], 
\label{ex50}
\end{eqnarray}
where $B_{K,j} = K\binom{N}{K} \binom{K-1}{j} (-1)^j $, $U_{K,j}=N-K+1+j$, $V_{K,j}=e^{-\lambda_s U_{K,j}(T_{\mathrm{D}}-c)}$, and 
$\mathcal{Z}_{K} = \mathbb{P}(T_{\mathrm{D}}<T_N(K))= \sum_{i=0}^{K-1} B_{K,i}\frac{V_{K,i}}{U_{K,i}}$.

\end{prop}
\begin{proof}
See Appendix A.
\end{proof}

The occurrence probability of Case $\mathcal{C}_{\mathrm{F},2}$ is given in the following proposition.
\begin{prop}
	The probability that Case $\mathcal{C}_{\mathrm{F},2}$ occurs can be expressed as
	\begin{equation}
		\begin{split}
			\mathbb{P}(\mathcal{C}_{\mathrm{F},2})=\frac{(N-K)\mathcal{Z}_K + \sum_{h=1}^{K}\mathcal{Z}_h}{Ne^{-\lambda_s (T_\mathrm{D}-c)}+(N-K) + \sum_{h=K+1}^{N}\mathcal{Z}_h},
			\label{Pcf2}
		\end{split}
	\end{equation}
	where $\mathcal{Z}_{h}$ is defined in Proposition 1. 
\end{prop}

\begin{proof}
	See Appendix B.
\end{proof}

By definition, we have $\mathbb{P} \left( \mathcal{C}_{\mathrm{F},1} \right) = 1 - \mathbb{P} \left( \mathcal{C}_{\mathrm{F},2} \right)$. 
By substituting (\ref{Eq_TNK1}), (\ref{Eq_TNK2}), and (\ref{Pcf2}) into 
(\ref{eqn_xnf7}) and (\ref{ex17}),
we obtain $\mathbb{E}\left[ X_{n}^{\mathrm{F}} \right]$ and $\mathbb{E} \left[ \left(X_n^{\mathrm{F}} \right)^2 \right]$.

%Fig. \ref{fig:flow} show that age evolution under all cases. On the one hand, we consider node $n$ fail to receive update, i.e., case $I$ occurs, $T_{\mathrm{D}} < T_{n,j}$ or $T_{N}(K) < T_{n,j} $. Then, if $T_{\mathrm{D}} < T_{N}(K)$ , i.e., case $I_a$ occurs, AoI adds $T_{\mathrm{D}}$. Otherwise, case $I_a$ occurs, AoI adds $T_{N}(K)$.
%On the other hand, we consider node $n$ receive update successfully, i.e., case $II$ occurs, $T_{n,j} \le T_{\mathrm{D}}$ and $T_{n,j} \le T_{N}(K)$. In this time AoI reset to a small value. Then, node $n$ still need to wait the deadline expires or remaining of $k$ nodes receive update. If $T_{\mathrm{D}} < T_{N}(K)$ , i.e., case $II_a$ occurs, AoI continues to adds $T_{\mathrm{D}}$. Otherwise, case $II_a$ occurs, AoI continues to adds $T_{N}(K)$.  By the way, if
%$T_{N}(K) = T_{n,j} $, node $n$ does not need to wait.

%\begin{figure}
%	\centering
%	\includegraphics[scale=0.5]{flow.eps}
%	\caption{AoI evolution with all cases.}
%	\label{fig:flow}
%\end{figure}

\subsection{First and Second Moments of Inter-Generation Time $X_n^{\mathrm{S}}$ for Fixed Deadline Case}
In this subsection, we derive the first and second moments of the inter-generation time of two consecutive status updates when the former status update is successfully received by device $n$, i.e., $\mathbb{E}\left[ X_n^{\mathrm{S}} \right]$ and $\mathbb{E}\left[ \left( X_n^{\mathrm{S}} \right)^2 \right]$.

Note that device $n$ successfully receives status update $j$ if $ T_{n,j} \le \min \{ T_{\mathrm{D}}, T_{N}(K) \}$.
We denote the case that device $n$ successfully receives the status update as $\mathcal{C}_{\mathrm{S}}$.
We observe that $X_{n}^{\mathrm{S}}$ behaves differently for the following two cases: (1) $\mathcal{C}_{\mathrm{S},1}$ - Event II occurs earlier than Event I (i.e., $T_{N}(K) < T_{\mathrm{D}}$); (2) $\mathcal{C}_{\mathrm{S},2}$ - Event I occurs earlier than Event II (i.e., $T_{\mathrm{D}} < T_N(K)$).
When Case $\mathcal{C}_{\mathrm{S},1}$ occurs, the instantaneous AoI of device $n$ increases by $T_{N}(K)$ (i.e., $X_n^{\mathrm{S}} = T_{N}(K)$).
When Case $\mathcal{C}_{\mathrm{S},2}$ occurs, the instantaneous AoI of device $n$ increases by $T_{\mathrm{D}}$ (i.e., $X_n^{\mathrm{S}} = T_{\mathrm{D}}$).
The first and second moments of $\mathrm{E}[X_{n}^\mathrm{S}]$ are given by
\begin{eqnarray}
\label{ex30_1} \mathbb{E}[X_{n}^\mathrm{S}] && \hspace{-6mm} =\mathbb{P}(\mathcal{C}_{\mathrm{S},1}) \mathbb{E}[T_N(K)|\mathcal{C}_{\mathrm{S},1}]+\mathbb{P}(\mathcal{C}_{\mathrm{S},2})T_{\mathrm{D}}, \\
\mathbb{E}\left[\left(X_{n}^\mathrm{S}\right)^2\right] &&\hspace{-6mm} =\mathbb{P}(\mathcal{C}_{\mathrm{S},1}) \mathbb{E}[T_N^2(K)|\mathcal{C}_{\mathrm{S},1}]+\mathbb{P}(\mathcal{C}_{\mathrm{S},2})T_{\mathrm{D}}^2,
\label{ex30}
\end{eqnarray}
where $\mathbb{P} \left( \mathcal{C}_{\mathrm{S},1} \right)$ and $\mathbb{P} \left( \mathcal{C}_{\mathrm{S},2} \right)$ denote the probabilities of the occurrence of Cases $\mathcal{C}_{\mathrm{S},1}$ and $\mathcal{C}_{\mathrm{S},2}$ when device $n$ successfully receives the status update, respectively, with $\mathbb{P} \left( \mathcal{C}_{\mathrm{S},1} \right) + \mathbb{P} \left( \mathcal{C}_{\mathrm{S},2} \right) = 1$.
To obtain $\mathbb{E}[X_{n}^\mathrm{S}]$ and $\mathbb{E}\left[\left(X_{n}^\mathrm{S}\right)^2\right]$, we need to calculate $\mathbb{P} \left( \mathcal{C}_{\mathrm{S},1} \right)$, $\mathbb{E}[T_N(K)|\mathcal{C}_{\mathrm{S},1}]$, and $\mathbb{E}[T_N^2(K)|\mathcal{C}_{\mathrm{S},1}]$.
The following proposition gives the first and second moments of $T_N(K)$ conditioning on the occurrence of Case $\mathcal{C}_{\mathrm{S},1}$.
\begin{prop}
The first and second moments of the time that $K$ IoT devices successfully receive a status update (i.e., $T_N(K)$) conditioning on the occurrence of Case $\mathcal{C}_{\mathrm{S},1}$ are given by
$\mathbb{E}[T_N(K)|\mathcal{C}_{\mathrm{S},1}] = \mathbb{E}\left[T_N(K) \left| \mathcal{C}_{\mathrm{F},1} \right. \right]$ and 
$\mathbb{E}[T_N^2(K)|\mathcal{C}_{\mathrm{S},1}] = \mathbb{E}\left[T_N^2(K) \left| \mathcal{C}_{\mathrm{F},1} \right. \right]$, 
where $\mathbb{E}\left[T_N(K) \left| \mathcal{C}_{\mathrm{F},1} \right. \right] $ and $\mathbb{E}\left[T_N^2(K) \left| \mathcal{C}_{\mathrm{F},1} \right. \right]$ are given in Proposition 1.
% \begin{eqnarray}
%   \mathbb{E} \left[ T_N(K) \left| \mathcal{C}_{\mathrm{S},1} \right. \right]&& \hspace{-6mm} =\int_{0}^{T_{\mathrm{D}}}tF'_{T_N(K)|\mathcal{C}_{\mathrm{S},1}}(t) \nonumber \\
%     && \hspace{-6mm} =\left(K\binom{N}{k} \sum_{j=0}^{K-1}\binom{K-1}{j} (-1)^j \frac{1-e^{-\lambda_s\mathcal{R}T_{\mathrm{D}}}(\lambda_s\mathcal{R}T_{\mathrm{D}}+1)}{\mathcal{R}^2}\right)/ \nonumber \\
%     && \hspace{-6mm} \left(1-K\binom{N}{K} \sum_{j=0}^{K-1}\binom{K-1}{j} (-1)^j \frac{e^{-\lambda_s\mathcal{R}T_{\mathrm{D}}}}{\mathcal{R}}\right), \\
%
% \mathbb{E} \left[ T_N^2(K) \left| \mathcal{C}_{\mathrm{S},1} \right. \right] && \hspace{-6mm} =  \int_{0}^{T_{\mathrm{D}}}t^2F'_{T_N(K)|\mathcal{C}_{\mathrm{S},1}}(t) \nonumber\\
%    && \hspace{-6mm} k\binom{N}{K} \sum_{j=0}^{K-1}\binom{K-1}{j} (-1)^j \nonumber\\
%    && \hspace{-6mm} \frac{2-e^{-\lambda_s\mathcal{R}T_{\mathrm{D}}}((\lambda_s\mathcal{R}T_{\mathrm{D}})^2+2\lambda_s\mathcal{R}T_{\mathrm{D}}+2)}{\lambda_s^2\mathcal{R}^3} /\nonumber\\
%   && \hspace{-6mm} \left(1-K\binom{N}{K} \sum_{j=0}^{K-1}\binom{K-1}{j} (-1)^j \frac{e^{-\lambda_s\mathcal{R}T_{\mathrm{D}}}}{\mathcal{R}}\right).
% \end{eqnarray}
\end{prop}

\begin{proof}
See Appendix C. 
\end{proof}
%  where $\Phi_2=K\binom{N}{K} \sum_{j=0}^{K-1}\binom{K-1}{j} (-1)^j \frac{2-e^{-\lambda_s\mathcal{R}T_{\mathrm{D}}}((\lambda_s\mathcal{R}T_{\mathrm{D}})^2+2\lambda_s\mathcal{R}T_{\mathrm{D}}+2)}{\lambda_s^2\mathcal{R}^3}$.

By definition, the occurrence probability of Case $\mathcal{C}_{\mathrm{S},1}$ is
\begin{equation}
\begin{split}
\mathbb{P}(\mathcal{C}_{\mathrm{S},1}) =  \frac{K(1 - \mathcal{Z}_{K})}{N\mathbb{P}(\mathcal{C}_{\mathrm{S}})},
\label{ex53}
\end{split}
\end{equation}
where $\mathbb{P}(\mathcal{C}_{\mathrm{S}})$ denotes the probability that device $n$ successfully receive update and is given by
\begin{equation}
	\begin{split}
\mathbb{P}(\mathcal{C}_{\mathrm{S}})&=\mathbb{P}( T_{n,j} < \min \{ T_{\mathrm{D}}, T_N(K)\})\\
& = \frac{1}{N} \sum_{h=1}^{K}   \left(1 - \mathcal{Z}_{h}\right).
\label{ex19}
\end{split}
\end{equation}

By substituting the derived expressions of $\mathbb{P}(\mathcal{C}_{\mathrm{S}})$, $\mathbb{E}[T_N(K) | C_{\mathrm{S},1}]$, and $\mathbb{E}[T^2_N(K) | C_{\mathrm{S},1}]$ into (15) and (16), we obtain $\mathbb{E}\left[X_n^{\mathrm{S}}\right]$ and $\mathbb{E}\left[\left(X_n^{\mathrm{S}}\right)^2\right]$.

\subsection{First and Second Moments of $W$ for Fixed Deadline Case}\label{w}
Recall that $W$ is the summation of $M-1$ consecutive inter-generation time $X_{n,j}^{\mathrm{F}}$, i.e., $W=\sum_{i=j}^{j+M-2}X_{n,i}^{\mathrm{F}}$.
As the probability that device $n$ successfully receives each status update is the same, $M$ is a geometric random variable.
As a result, the probability mass function (PMF) of $M$ is given by $\mathbb{P}(M = m)=(1-\mathbb{P}(\mathcal{C}_{\mathrm{S}}))^{m-1}\mathbb{P}(\mathcal{C}_{\mathrm{S}}), m \ge 1$, where $\mathbb{P}(\mathcal{C}_{\mathrm{S}})$ is given in \eqref{ex19}.
Obviously, we have $\mathbb{E}[M]={1}/{\mathbb{P}(\mathcal{C}_{\mathrm{S}})}$ and $\mathbb{E}[M^2]=\frac{2-\mathbb{P}(\mathcal{C}_{\mathrm{S}})}{\mathbb{P}(\mathcal{C}_{\mathrm{S}})^2}$.
As $M$ and $X_{n}^\mathrm{F}$ are independent, the first moment of $W$ can be calculated by 
 \begin{equation}
 \mathbb{E}[W]=(\mathbb{E}[M]-1)\mathbb{E}[X_{n}^\mathrm{F}].
 \label{ex20}
 \end{equation}

To derive the expression of $\mathbb{E}[W^2]$, we first calculate the variance of $W$, which is given by
\begin{eqnarray}
 \mathrm{Var}[W] && \hspace{-6mm} =\mathrm{Var}\left[\mathbb{E}[W|M]\right]+\mathbb{E}\left[\mathrm{Var}[W|M]\right] \nonumber \\
 && \hspace{-6mm} =\left(\mathbb{E}\left[X_n^{\mathrm{F}}\right]\right)^2 \mathrm{Var}[M]+\mathrm{Var}[X_n^{\mathrm{F}}]\left(\mathbb{E}[M]-1\right), 
 \label{ex21}
\end{eqnarray}
where $\mathrm{Var}\left[ X_{n}^{\mathrm{F}} \right] = \mathbb{E}\left[ \left(X_n^{\mathrm{F}} \right)^2 \right] - \left( \mathbb{E}\left[ X_{n}^{\mathrm{F}} \right] \right)^2$. 
Based on \eqref{ex20} and \eqref{ex21}, we obtain the second moments of $W$ given by $\mathbb{E}[W^2]=(\mathbb{E}[W])^2+\mathrm{Var}[W]$.

\subsection{First Moment of Successful Service Time  $\hat{T}_{n,q}$  for Fixed Deadline Case}\label{SubSec_TNQ}
Recall that $\hat{T}_{n,q}$ is the service time of the $q$-th status update successfully delivered to device $n$.
Conditioning on the occurrence of Case $\mathcal{C}_{\mathrm{S}}$, the CDF of the service time is
\begin{equation}
	\begin{split}
		& F_{T_{n,j} | \mathcal{C}_{\mathrm{S}}}(t) =  \mathbb{P}( T_{n,j}<t | \mathcal{C}_{\mathrm{S}})  \\
		&= \frac{1}{N \mathbb{P}(\mathcal{C}_{\mathrm{S}})} \sum_{h=1}^{K}\left( 1-  \mathcal{Z}_{h} -  \sum_{j=0}^{h-1} B_{h,j} \frac{e^{-\lambda_s(t-c)U_{h,j}}-V_{K,j}}{{U_{h,j}}}\right),
		 \label{ex23}
\end{split}
\end{equation}
where $U_{h,j}$, $V_{K,j}$, $B_{h,j}$, and $\mathcal{Z}_{h}$ are defined in Proposition 1, and $\mathbb{P}(\mathcal{C}_{\mathrm{S}})$ is given in (\ref{ex19}). 
Based on (\ref{ex23}), the expectation of $\hat{T}_{n,q}$ can be calculated by 
\begin{equation}
	\begin{split}
		& \mathbb{E}[\hat{T}_{n,q}] =\int_{c}^{T_\mathrm{D}}t \,\,\, \mathrm{d}\, F_{T_{n,j} | \mathcal{C}_{\mathrm{S}}}(t) \\
		=& \sum_{h=1}^{K}\sum_{j=0}^{h-1} B_{h,j}  \frac{c\lambda_sU_{h,j}+1-e^{-\lambda_sU_{h,j}T_\mathrm{D}}(\lambda_sU_{h,j}T_\mathrm{D}+1)}{N\mathbb{P}(\mathcal{C}_{\mathrm{S}}) \lambda_s U_{h,j}^2}. 
	  \label{ex24}
\end{split}
\end{equation}
%where $f_{T_{n,j} | \mathcal{C}_{\mathrm{S}}}(t)$ is the first derivative of $F_{T_{n,j} | \mathcal{C}_{\mathrm{S}}}(t)$.

% \begin{eqnarray}
%  \mathbb{E}[\hat{T}_{n,q}]&& \hspace{-6mm} =\int_{0}^{T_{\mathrm{D}}}t f_{T_{n,j} | \mathcal{C}_{\mathrm{S}}}(t) \mathrm
% {d} t  \nonumber \\
%  && \hspace{-6mm} =\frac{h}{N\mathbb{P}(\mathcal{C}_{\mathrm{S}})}\sum_{h=1}^{K} \binom{N}{h} \sum_{j=0}^{h-1}\binom{h-1}{j} (-1)^j \frac{1-e^{-\lambda_s\mathcal{R}T_{\mathrm{D}}}(\lambda_s\mathcal{R}T_{\mathrm{D}}+1)}{\lambda_s\mathcal{R}^2},
%  \label{ex24}
% \end{eqnarray}
%where $f_{T_{n,j} | \mathcal{C}_{\mathrm{S}}}(t)$ denotes the first derivative of $F_{T_{n,j} | \mathcal{C}_{\mathrm{S}}}(t)$.

 \subsection{Average (Peak) AoI for Fixed Deadline Case}
Based on the above analysis, we obtain the average AoI of the multicast transmission with fixed deadlines by substituting \eqref{ex30_1}, \eqref{ex30}, \eqref{ex20}, and \eqref{ex24} into \eqref{AveAoI}.
Similarly, we obtain the corresponding average peak AoI by substituting \eqref{ex30_1}, \eqref{ex20}, and \eqref{ex24} into \eqref{PeakAoI}. 
It is worth pointing out that the results presented in this paper can be easily extended to the scenarios for broadcast transmission with deadlines by replacing $K$ with $N$, for multicast transmission without deadlines by setting $T_{\mathrm{D}} = \infty$, and for unicast transmission with deadlines by setting $N = K =1$.

\section{Analysis of Average (Peak) AoI with Randomly Distributed Deadlines}\label{random}
In this section, we derive the average (peak) AoI of multicast transmission with randomly distributed deadlines. 
Compared to fixed deadlines, the performance analysis for random deadlines is further complicated as the impact of the distribution of random deadlines on the evolution of the instantaneous AoI at each device should be taken into account. 
Hence, all terms in the expressions of the average (peak) AoI need to be recalculated. 
In particular, the distribution of the random deadlines not only determines the occurrence probabilities of both successful and failed status update receptions at each device, but also the distribution of the inter-generation time for each reception outcome.
Studying the random deadline case helps understanding the AoI performance of the status update systems where different status updates are subject to different deadlines. 
Recall that the deadline associated with status update $j$ is $T_{\mathrm{D}, j}$, which is assumed to follow an exponential distribution with rate $\lambda_d$ and constant shift $c$. 
Hence, the PDF of deadline $T_{\mathrm{D}, j}$ is given by  $f_{T_{\mathrm{D}, j}}(t)=\lambda_d \mathrm{e}^{-\lambda_d (t-c)}, t > c, \forall \, j \in \mathcal{J}$. 
By denoting $\mathcal{X}_n^{\mathrm{F}}$, $\mathcal{X}_{n}^\mathrm{S}$, $\mathcal{M}_{n,q}$, $\mathcal{W}_{n,q}$, and $\hat{\mathcal{T}}_{n,q}$ for the random deadline case as the counterparts of ${X}_n^{\mathrm{F}}$, $X_{n}^{S}$, $M_{n,q}$, $W_{n,q}$, and $\hat{T}_{n,q}$ for the fixed deadline case respectively, we can rewrite the average AoI in \eqref{AveAoI} and the average peak AoI in \eqref{PeakAoI} as 
\begin{eqnarray}
\label{AveAoI_1} \widetilde{\Delta}_n && \hspace{-6mm}= \frac{\mathbb{E}\left[\mathcal{W}^2\right]  + 2 \left( \mathbb{E}\left[\hat{\mathcal{T}}_{n,q}\right] + \mathbb{E}\left[\mathcal{X}_{n}^\mathrm{S}\right] \right)
		\mathbb{E}[\mathcal{W}]  }{2\mathbb{E}[\mathcal{W}]+
		2\mathbb{E}\left[\mathcal{X}_{n}^\mathrm{S}\right]} \nonumber \\
	  && \hspace{-3mm} + \, \frac{2 \mathbb{E}\left[\mathcal{X}_{n}^\mathrm{S}\right]\mathbb{E}\left[\hat{\mathcal{T}}_{n,q}\right]
		 + \mathbb{E}\left[{\left(\mathcal{X}_{n}^\mathrm{S}\right)}^2\right]  }{2\mathbb{E}[\mathcal{W}]+2\mathbb{E}\left[\mathcal{X}_{n}^\mathrm{S}\right]}, \\
\label{PeakAoI_1}		 \widetilde{P}_n && \hspace{-6mm} = \mathbb{E}[\mathcal{X}_n^S] + \mathbb{E}[\mathcal{W}] + \mathbb{E}[\hat{\mathcal{T}}_{n,q}].  
\end{eqnarray}

In the following subsections, we derive the closed-form expressions of all the expectation terms in \eqref{AveAoI_1} and \eqref{PeakAoI_1}. 

\subsection{First and Second Moments  of Inter-Generation Time $\mathcal{X}_n^{\mathrm{F}}$ for Random Deadline Case}
%We first calculate the expectation of the inter-generation time of two consecutive status updates when the former status update fails to be delivered to device $n$, i.e., $\mathbb{E}[\mathcal{X}_{n}^\mathrm{F}]$. 
With randomly distributed deadlines for status updates, we need to rederive the first and second moments of the inter-generation time of two consecutive status updates when the former status update fails to be delivered to device $n$, i.e., $\mathbb{E}[\mathcal{X}_{n}^\mathrm{F}]$ and $\mathbb{E}[\left( \mathcal{X}_{n}^\mathrm{F} \right)^2]$. 
The case that device $n$ fails to receive the status update, i.e., Case $\mathcal{C}_{\mathrm{F}}$, occurs when $T_{n,j} > \min\{T_{\mathrm{D}, j}, T_{N}(K) \}$. 
Case $\mathcal{C}_{\mathrm{F}}$ can further be divided into two cases, i.e., Cases $\mathcal{C}_{\mathrm{F},1}$ and $\mathcal{C}_{\mathrm{F},2}$, which occur if $T_{N}(K) < T_{\mathrm{D}, j}$ and $T_{\mathrm{D}, j}  < T_{N}(K)$, respectively. 
Thus, the first and second moments of inter-generation time $\mathcal{X}_n^F$ are given by
\begin{equation}
	\begin{split}
\mathbb{E}\left[ \mathcal{X}_{n}^{\mathrm{F}} \right] = \mathbb{P} \left( \mathcal{C}_{\mathrm{F},1}  \right) \mathbb{E}\left[T_N(K) \left| \mathcal{C}_{\mathrm{F},1}\!\right.\right] + \mathbb{P} \left( \mathcal{C}_{\mathrm{F},2}  \right) \mathbb{E}\left[ T_{\mathrm{D}, j}| \mathcal{C}_{\mathrm{F},2} \right],
\label{eqn_xnf7r}
\end{split}
\end{equation}
\begin{equation}
	\begin{split}
		\mathbb{E} \left[ \left(\mathcal{X}_n^{\mathrm{F}} \right)^2 \right] = \mathbb{P} \left( \mathcal{C}_{\mathrm{F},1} \right) \mathbb{E}[T_N^2(K)| \mathcal{C}_{\mathrm{F},1} ] + \mathbb{P} \left( \mathcal{C}_{\mathrm{F},2}  \right) \mathbb{E}\left[ T_{\mathrm{D}, j}^2| \mathcal{C}_{\mathrm{F},2} \right], 
 \label{ex17rb}
\end{split}
\end{equation}
where $\mathbb{P}(\mathcal{C}_{\mathrm{F},1})$ and $\mathbb{P}(\mathcal{C}_{\mathrm{F},2})$ denote the occurrence probabilities of Cases $\mathcal{C}_{\mathrm{F},1}$ and $\mathcal{C}_{\mathrm{F},2}$ in the random deadline case, respectively. 

Based on the definition of Case $\mathcal{C}_{\mathrm{F},2}$, we have
\begin{equation}
	\begin{split}
		\mathbb{P}(\mathcal{C}_{\mathrm{F},2}) & =  \,\mathbb{P}( T_{\mathrm{D}, j}<T_N(K)|T_{n,j} > \min\{T_{\mathrm{D}, j}, T_{N}(K) \} ) \\
		& = \frac{(\sum_{h=1}^{K}\mathcal{R}_h+(N-K)\mathcal{R}_K)(\lambda_s+\lambda_d)}{N\lambda_d+(N-K-\sum_{h=K+1}^{N}\mathcal{R}_h)(\lambda_s+\lambda_d)},
		\label{eq63r}
	\end{split}
\end{equation}
where $\mathcal{R}_K=\mathbb{P}( T_N(K)\le T_{\mathrm{D}, j})=\sum_{j=0}^{K-1}B_{K,j}(\frac{1}{U_{K,j}}-\frac{\lambda_s}{H_{K,j}})$, $B_{K,j}$ and $U_{K,j}$ are defined in Proposition 1, and $H_{K,j}=\lambda_s U_{K,j}+\lambda_d$. 
Meanwhile, we obtain $\mathbb{P}(\mathcal{C}_{\mathrm{F},1})=1-\mathbb{P}(\mathcal{C}_{\mathrm{F},2})$. 

The AoI of device $n$ keeps increasing before successfully receiving a fresher status update.  
On one hand, the instantaneous AoI increases by $T_N(K)$ if $T_N(K)<T_{\mathrm{D}, j}$.
The first and second moments of $T_{N}(K)$ conditioning on the occurrence of Case $\mathcal{C}_{\mathrm{F},1}$ are presented in Proposition \ref{Prop_TNKCF1}.

\begin{prop} \label{Prop_TNKCF1}
The first and second moments of the time required for $K$ devices to successfully receive a status update (i.e., $T_{N}(K)$) conditioning on the occurrence of Case $\mathcal{C}_{\mathrm{F},1}$ can be calculated by 
\begin{eqnarray}
\label{eqn8r}   &\mathbb{E}\left[T_N(K) \left| \mathcal{C}_{\mathrm{F},1} \right. \right] =\sum_{j=0}^{K-1}B_{K,j}\lambda_s\frac{cNH_{K,j}+1}{H_{K,j}^2(N-K)(1-\mathcal{R}_K)}, \\
 \label{eqn9r}  &\mathbb{E}\left[T_N^2(K) \left| \mathcal{C}_{\mathrm{F},1} \right. \right] =\sum_{j=0}^{K-1}B_{K,j}\lambda_s\frac{c^2 N H_{K,j}^2 +2cH_{K,j}+2}{H_{K,j}^3(N-K)(1-\mathcal{R}_K)}.
\label{ex50r}
\end{eqnarray}
\end{prop}
\begin{proof}
See Appendix D.
\end{proof}

On the other hand, if $T_{\mathrm{D}, j}<T_N(K)$, then the instantaneous AoI increases by $T_{\mathrm{D}, j}$, the first and second moments of which conditioning on the occurrence of Case $\mathcal{C}_{\mathrm{F},2}$ are given in Proposition \ref{Prop_TDJCF2}. 

\begin{prop} \label{Prop_TDJCF2}
The first and second moments of $T_{\mathrm{D}, j}$ conditioning on the occurrence of Case $\mathcal{C}_{\mathrm{F},2}$ can be calculated by 
\begin{equation}
	\begin{split}
		&\mathbb{E}\left[T_{\mathrm{D}, j} \left| \mathcal{C}_{\mathrm{F},2} \right. \right]
		=\left(\sum_{h=1}^{K}\sum_{j=0}^{h-1}B_{h,j}\lambda_d\frac{cH_{h,j}+1}{U_{h,j}H_{h,j}^2}\right. \\
		&\left. + N\sum_{j=0}^{K-1}B_{K,j}\lambda_d\frac{cH_{K,j}+1}{U_{K,j}H_{K,j}^2}\right)\frac{1}{\sum_{h=1}^{K}\mathcal{R}_h+(N-K)\mathcal{R}_K},
		\label{eqn8rb}
	\end{split}
\end{equation}
\begin{equation}
	\begin{split}
		&\mathbb{E}\left[T_{\mathrm{D}, j}^2 \left| \mathcal{C}_{\mathrm{F},2} \right. \right] 
		=\left(\sum_{h=1}^{K}\sum_{j=0}^{h-1}B_{h,j}\lambda_d\frac{c^2H_{h,j}^2+2cH_{h,j}+2}{U_{h,j}H_{h,j}^3}\right.\\
		&\left.+N\sum_{j=0}^{K-1}B_{K,j}\lambda_d\frac{c^2H_{K,j}^2+2cH_{K,j}+2}{U_{K,j}H_{K,j}^3}\right)\frac{1}{\sum_{h=1}^{K}\mathcal{R}_h \!+\! (N-K)\mathcal{R}_K},
		\label{eqn9rb}
	\end{split}
\end{equation}

\end{prop}
\begin{proof}
See Appendix E.
\end{proof}

By substituting \eqref{eq63r}--\eqref{eqn9rb} and $\mathbb{P}(\mathcal{C}_{\mathrm{F},1})$ into
\eqref{eqn_xnf7r} and \eqref{ex17rb}, we obtain $\mathbb{E} \left[ \mathcal{X}_n^{\mathrm{F}} \right]$ and $\mathbb{E} \left[ \left(\mathcal{X}_n^{\mathrm{F}} \right)^2 \right]$.

\subsection{First and Second Moments of Inter-Generation Time $\mathcal{X}_n^{\mathrm{S}}$ for Random Deadline Case}
In this subsection, we calculate the first and second moments of the inter-generation time of two consecutive status updates when the former status update is successfully received by device $n$ for the random deadline case, i.e., $\mathbb{E}\left[ \mathcal{X}_n^{\mathrm{S}} \right]$ and $\mathbb{E}\left[ \left( \mathcal{X}_n^{\mathrm{S}} \right)^2\right]$.
If $T_{n,j} \le \min \{ T_{\mathrm{D}, j}, T_{N}(K) \}$, then Case $\mathcal{C}_{\mathrm{S}}$ occurs and can be further categorized into Cases $\mathcal{C}_{\mathrm{S},1}$ and $\mathcal{C}_{\mathrm{S},2}$, which occur when $T_{N}(K) < T_{\mathrm{D}, j}$ and $T_{\mathrm{D}, j} < T_N(K)$, respectively. 
As a result, the first and second moments of inter-generation time $\mathcal{X}_n^S$ can be expressed as
\begin{eqnarray}
\hspace{-5mm} \label{ex30r} \mathbb{E}[\mathcal{X}_{n}^\mathrm{S}] && \hspace{-6mm} =\mathbb{P}(\mathcal{C}_{\mathrm{S},1}) \mathbb{E}[T_N(K)|\mathcal{C}_{\mathrm{S},1}]+\mathbb{P}(\mathcal{C}_{\mathrm{S},2})\mathbb{E}[T_{\mathrm{D}, j}|\mathcal{C}_{\mathrm{S},2}], \\
\hspace{-5mm}  \mathbb{E}\left[\left(\mathcal{X}_{n}^\mathrm{S}\right)^2\right] &&\hspace{-6mm} =\mathbb{P}(\mathcal{C}_{\mathrm{S},1}) \mathbb{E}[T_N^2(K)|\mathcal{C}_{\mathrm{S},1}]+\mathbb{P}(\mathcal{C}_{\mathrm{S},2})\mathbb{E}[T_{\mathrm{D}, j}^2|\mathcal{C}_{\mathrm{S},2}],
\label{ex31r}
\end{eqnarray}
where $\mathbb{P} \left( \mathcal{C}_{\mathrm{S},1} \right)$ and $\mathbb{P} \left( \mathcal{C}_{\mathrm{S},2} \right)$ denote the probabilities of the occurrence of Cases $\mathcal{C}_{\mathrm{S},1}$ and $\mathcal{C}_{\mathrm{S},2}$ when device $n$ successfully receives the status update, respectively.
Base on the definition of Case $\mathcal{C}_{\mathrm{S},1}$, we have
\begin{equation}
	\begin{split}
		\mathbb{P}(\mathcal{C}_{\mathrm{S},1})&=\mathbb{P}( T_N(K)<T_{\mathrm{D}, j}| T_{n,j}<\min\{ T_{\mathrm{D}, j}, T_N(K)\}) \\
		&=\frac{K \mathbb{P}( T_N(K)\le T_{\mathrm{D}, j})}{\sum_{h=1}^{K}\mathbb{P}( T_N(h)\le T_{\mathrm{D}, j})} \\
		&=\frac{K\mathcal{R}_K}{\sum_{h=1}^{K}\mathcal{R}_h}.
		\label{eq65r}
	\end{split}
\end{equation}
We can easily obtain the probability of $\mathcal{C}_{\mathrm{S},2}$ as $\mathbb{P}(\mathcal{C}_{\mathrm{S},2})=1-\mathbb{P}(\mathcal{C}_{\mathrm{S},1})$.
If device $n$ successfully receives the status update, then the instantaneous AoI is reset to $T_{\mathrm{D}, j}$ when Case $\mathcal{C}_{\mathrm{S},2}$ occurs. 
In this case, the first and second moments of $T_{\mathrm{D}, j}$ in \eqref{ex30r} and \eqref{ex31r} are provided in the following proposition. 

\begin{prop}
The first and second moments of $T_{\mathrm{D}, j}$ conditioning on the occurrence of Case $\mathcal{C}_{\mathrm{S},2}$ can be calculated by 
\begin{equation}
	\begin{split}
		&\mathbb{E}\left[T_{\mathrm{D}, j} \left| \mathcal{C}_{\mathrm{S},2} \right. \right] 
		= \frac{1}{K\mathcal{R}_K - \sum_{h=1}^{K}\mathcal{R}_h} \left(K\sum_{j=0}^{K-1}B_{K,j}\lambda_d\frac{cH_{K,j}+1}{U_{K,j}H_{K,j}^2}\right. \\
		& \hspace{22mm} \left. - \sum_{h=1}^{K}\sum_{j=0}^{h-1}B_{h,j}\lambda_d\frac{cH_{h,j}+1}{U_{h,j}H_{h,j}^2} \right),
		\label{eqn8rrb}
	\end{split}
\end{equation}
\begin{equation}
	\begin{split}
		&\mathbb{E}\left[T_{\mathrm{D}, j}^2 \left| \mathcal{C}_{\mathrm{S},2} \right. \right] 
		= \frac{1}{K\mathcal{R}_K - \sum_{h=1}^{K}\mathcal{R}_h} \left(K\sum_{j=0}^{K-1}B_{K,j}\lambda_d \right. \\
		&\left. \times \frac{c^2H_{K,j}^2+2cH_{K,j}+2}{U_{K,j}H_{K,j}^3} - \sum_{h=1}^{K}\sum_{j=0}^{h-1}B_{h,j}\lambda_d\frac{c^2H_{h,j}^2 \!+\! 2cH_{h,j}+2}{U_{h,j}H_{h,j}^3} \right).
		\label{eqn9rrb}
	\end{split}
\end{equation}

\end{prop}

\begin{proof}
See Appendix F.
\end{proof}

If device $n$ successfully receives the status update, then the instantaneous AoI is reset to $T_N(K)$ when Case $\mathcal{C}_{\mathrm{S},1}$ occurs. 
In this case, the first and second moments of $T_{N}(K)$ in \eqref{ex30r} and \eqref{ex31r} are presented in Corollary \ref{Coro_TNK}.

\begin{corollary} \label{Coro_TNK}
The first and second moments of $T_N(K)$ conditioning on the occurrence of Case $\mathcal{C}_{\mathrm{S},1}$ is given by
$\mathbb{E}\left[T_N(K) \left| \mathcal{C}_{\mathrm{S},1} \right. \right]=\mathbb{E}\left[T_N(K) \left| \mathcal{C}_{\mathrm{F},1} \right. \right]$ and $\mathbb{E}\left[T_N^2(K) \left| \mathcal{C}_{\mathrm{S},1} \right. \right] = \mathbb{E}\left[T_N^2(K) \left| \mathcal{C}_{\mathrm{F},1} \right. \right]$, respectively, where $\mathbb{E}\left[T_N(K) \left| \mathcal{C}_{\mathrm{F},1} \right. \right]$ and $\mathbb{E}\left[T_N^2(K) \left| \mathcal{C}_{\mathrm{F},1} \right. \right]$ are given in Proposition 4.
\end{corollary}

By substituting the derived $\mathbb{E}\left[T_{\mathrm{D}, j} \left| \mathcal{C}_{\mathrm{S},2} \right. \right]$, $\mathbb{E}[T_{\mathrm{D}, j}^2|\mathcal{C}_{\mathrm{S},2}]$, $\mathbb{E}\left[T_N(K) \left| \mathcal{C}_{\mathrm{S},1} \right. \right]$, $\mathbb{E}\left[T_N^2(K) \left| \mathcal{C}_{\mathrm{S},1} \right. \right]$, $\mathbb{P}(\mathcal{C}_{\mathrm{S},1})$, and $\mathbb{P}(\mathcal{C}_{\mathrm{S},2})$ into \eqref{ex30r} and  \eqref{ex31r}, we obtain the first and second moments of inter-generation time $\mathcal{X}_n^{\mathrm{S}}$ for the random deadline case, i.e., $\mathbb{E}[\mathcal{X}_{n}^\mathrm{S}]$ and  $\mathbb{E}\left[\left(\mathcal{X}_{n}^\mathrm{S}\right)^2\right]$.

 \subsection{First and Second Moments of $\mathcal{W}$ for Random Deadline Case}\label{rw}
Similar to Section \ref{w},  $\mathcal{W}$ is the summation of $\mathcal{M}-1$ consecutive inter-generation time $\mathcal{X}_{n,j}^{\mathrm{F}}$, i.e., $\mathcal{W}=\sum_{i=j}^{j+\mathcal{M}-2}\mathcal{X}_{n,i}^{\mathrm{F}}$. 
Recall the definition of $\mathcal{M}$, we have $\mathbb{E}[\mathcal{M}]=\frac{1}{\mathbb{P}(\mathcal{C}_{\mathcal{S}})}$ and $\mathbb{E}[\mathcal{M}^2]=\frac{2-\mathbb{P}(\mathcal{C}_{\mathrm{S}})}{\mathbb{P}(\mathcal{C}_{\mathcal{S}})^2}$, where $\mathbb{P}(\mathcal{C}_{\mathcal{S}})=\frac{1}{N}\sum_{h=1}^{k}(1-\mathcal{R}_h)$.
 As $\mathcal{M}$ and $\mathcal{X}_{n}^\mathrm{F}$ are independent, the first moment of $\mathcal{W}$ is given by $\mathbb{E}[\mathcal{W}]=(\mathbb{E}[\mathcal{M}]-1)\mathbb{E}[\mathcal{X}_{n}^\mathrm{F}]$.
 Meanwhile, the second moment of $\mathcal{W}$ can be calculated by
 \begin{equation}
 	\begin{split}
 		\mathbb{E}[\mathcal{W}^2]&=\mathbb{E}[\mathrm{Var}(\mathcal{W}|\mathcal{M})]+\mathrm{Var}(\mathbb{E}[\mathcal{W}|\mathcal{M}])+\mathbb{E}^2[\mathcal{W}] \\
 		&=\mathbb{E}[\mathcal{M}-1]\mathrm{Var}\left[ \mathcal{X}_{n}^{\mathrm{F}} \right] +
		\mathbb{E}\left[ \mathcal{X}_{n}^{\mathrm{F}} \right]^2\mathrm{Var}(\mathcal{M})+\mathbb{E}[\mathcal{W}]^2,
 		\label{eq75r}
 	\end{split}
 \end{equation}
where
 $\mathrm{Var}\left[ \mathcal{X}_{n}^{\mathrm{F}} \right] = \mathbb{E}\left[ \left(\mathcal{X}_n^{\mathrm{F}} \right)^2 \right] - \left( \mathbb{E}\left[ \mathcal{X}_{n}^{\mathrm{F}} \right] \right)^2$.

 \subsection{First Moment of Successful Service Time  $\hat{\mathcal{T}}_{n,q}$ for Random Deadline Case}
In this subsection, we calculate the expectation of successful service time $\hat{T}_{n,q}$ conditioning on the occurrence of Case $\mathcal{C}_{\mathcal{S}}$ for the random deadline case. The CDF of the conditional service time is given by
\begin{equation}
	\begin{split}
		&F_{T_{n,j}|\mathcal{C}_{\mathcal{S}} }(t)=\mathbb{P}( T_{n,j}<t| T_{n,j}\le T_{\mathrm{D}, j}, T_{n,j}\le T_N(K)) \\
		&=\frac{\sum_{h=1}^{K}(1-\mathcal{R}_h)-\sum_{h=1}^{k}\mathbb{P}(t\le T_N(h)\le T_{\mathrm{D}, j})}{N \mathbb{P}(\mathcal{C}_{\mathcal{S}})}.
		\label{eq73r}
	\end{split}
\end{equation}
where  $\mathbb{P}(\mathcal{C}_{\mathcal{S}}) = \frac{1}{N}\sum_{h=1}^{k}(1-\mathcal{R}_h)$ and $\mathbb{P}(t\le T_N(h)\le T_{\mathrm{D}, j}) $ is given in \eqref{eq48r}. 
Thus, the expectation of successful service time  $\hat{\mathcal{T}}_{n,q}$ can be calculated by
\begin{equation}
	\begin{split}
		\mathbb{E}[\hat{\mathcal{T}}_{n,q}] & =\mathbb{E}[T_{n,j}|\mathcal{C}_{\mathcal{S}}]=\int_{c}^{+\infty} t \, \mathrm{d} F_{T_{n,j}|\mathcal{C}_{\mathcal{S}} }(t) \\
		& =\frac{1}{\mathbb{P}(\mathcal{C_S})}\sum_{h=1}^{K}\sum_{j=0}^{h-1}B_{h,j}\lambda_s\frac{cH_{h,j}+1}{H_{h,j}^2}.
		\label{eq74r}
	\end{split}
\end{equation}

\subsection{Average (Peak) AoI for Random Deadline Case}
Based on the aforementioned analysis, we obtain the average AoI and the average peak AoI of the multicast transmission with exponentially distributed deadlines by substituting \eqref{eqn_xnf7r}, \eqref{ex17rb}, \eqref{ex30r}, \eqref{ex31r}, \eqref{eq75r}, and \eqref{eq74r} into \eqref{AveAoI_1} and (\ref{PeakAoI_1}), respectively.

\section{Performance Evaluation and Discussions}\label{sim}
In this section, we present both the simulation and theoretical results in terms of the average (peak) AoI for multicast transmission with deadlines in IoT networks, and compare the results with that of the scenario without deadlines. 
We conduct Monte-Carlo simulations using MATLAB to verify the correctness of our theoretical analysis. 
The transmission process of $100,000$ consecutive status updates is simulated to obtain the instantaneous AoI evolution, which is then used to calculate the average (peak) AoI.
For performance comparison, we set the average deadline of the random deadline case (i.e., ${1}/{\lambda_d}+c$) to be the same as deadline $T_\mathrm{D}$ of the fixed deadline case. 

 \begin{figure}
   \centering
   \includegraphics[scale=0.58]{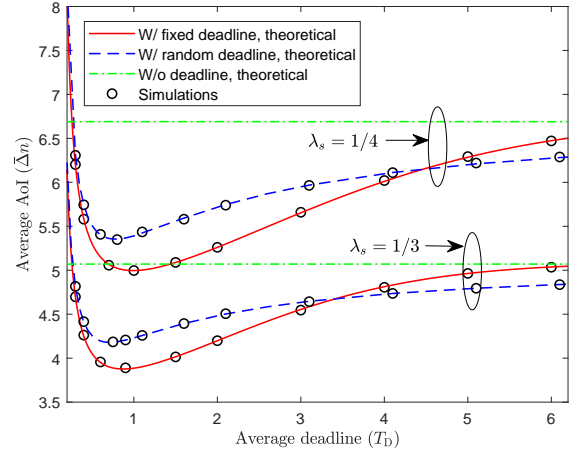}
   \caption{Average AoI versus average deadline $T_\mathrm{D}$ for different values of $\lambda_s$ when $K=7$, $N=10$, and $c=0.1$.}
   \label{fig:s}
   \vspace{-5mm}
 \end{figure}

	Fig. \ref{fig:s} shows the impact of deadline $T_{\mathrm{D}}$ on the average AoI for different values of average service rate $\lambda_s$ when $K=7$, $N=10$, and $c=0.1$. 
	As can be observed, the simulation and theoretical results match well, which validating the accuracy of the performance analysis in Sections III and IV. 
	For both fixed and random deadlines, with the variation of deadline $T_{\mathrm{D}}$, the average AoI first decreases to a minimum value and then increases to a saturation value. 
	By using the ternary search algorithm, we are able to numerically find the optimal value of the deadline that minimizes the average (peak) AoI.
	We take the fixed deadline case as an example. 
	When $\lambda_s = 1/3$ and deadline $T_{\mathrm{D}}$ is small, the probability that each device can successfully receive a status update within a transmission interval (i.e., $\min \{T_N(K), T_{\mathrm{D}}\}$) is also small. 
	As such, it may take each IoT device many transmission intervals to successfully receive a status update. 
 	Note that the average AoI is proportional to the average number of transmission intervals required to successfully receive a status update as well as the average length of transmission intervals.
 	Hence, the average AoI of the considered system is large when the deadline is small (e.g., 0.2). 
 	By increasing the value of average deadline $T_{\mathrm{D}}$ to $0.9$, the average AoI declines quickly until reaching its minimum value. 
 	This is due to the fact that the probability of successful status update reception within each transmission interval increases. 
 	By further increasing the value of deadline $T_{\mathrm{D}}$, the average length  of transmission intervals increases and it starts to play a more important role in the AoI evolution than the average number of transmission intervals required to successfully a status update, leading to the increase of the average AoI. 
 	When deadline $T_{\mathrm{D}}$ is sufficiently large, the average AoI approaches a saturation value and does not further vary with deadline $T_{\mathrm{D}}$. 
 	This corresponds to the case of multicast transmission without deadlines. 
 	In addition, we can also observe that the average AoI decreases as the value of $\lambda_s$ increases. 
	This is because a larger average service rate leads to a smaller average length of transmission intervals. 
	
	Fig. \ref{fig:s} also illustrates the average AoI comparison between the fixed and random deadline cases. 
	As can be observed, the minimum value of the average AoI for the fixed deadline case is smaller than that for the random deadline case. 
	This is because the fixed deadline case reduces the variability and limits the maximum possible value of the deadline. 
	The maximum possible value of the instantaneous deadline in the random deadline case can be very large with a certain probability, which has a detrimental effect on reducing the average AoI.
	%	This is because the fixed deadline case limits the maximum possible deadline, which has a detrimental effect on the average AoI. 
	%	However, the maximum possible deadline of the random deadline case can be very large with a certain probability.  
	As a result, this illustrates the importance of limiting the maximum possible deadline in reducing the average AoI. 
	For the random deadline case, some status updates have larger deadlines and other status updates have smaller deadlines when compared with the fixed deadline case. 
	When $T_{\mathrm{D}}$ is around its optimal value, the detrimental effect of status updates with larger deadlines cannot be mitigated by the status updates with smaller deadlines, and hence the random deadline case achieves a larger average AoI than the fixed deadline case. 
	On the other hand, when $T_{\mathrm{D}}$ is large, the beneficial effect due to status updates with smaller deadlines plays a dominating role in reducing the average AoI, while the detrimental effect of status updates with larger deadlines is negligible as most packets can be served before the deadline expires. 
	As a result, when $T_{\mathrm{D}}$ is large, the random deadline case achieves a better performance than that the fixed deadline case. 
	 
% Then, we analyze the difference between the fixed case and random case.
% We observe that average AoI in random case is greater than fixed case when deadline is small. And as $\lambda_s$ decreases, the range in which average AoI in random case is greater than fixed case becomes smaller. 
% This is because different deadline has different effect on average AoI. As mentioned above, small deadline has strong effect on average AoI than large deadline. In the long run, the average AoI is equal to fixed deadline. Actually in random case the server is subjected to  different deadline at each transmission. Thus, when deadlines is small, average AoI in random case is greater than fixed case and the range becomes smaller as $\lambda_s$ decreases.
% Meanwhile, as can be seen, the minimum average AoI point is different between two case. Obviously, $T_{\mathrm{D}}$ is smaller when the minimum point in random case is reached. This is also due to the fact that the server is subjected to different deadline at each transmission.
% Moreover, we can also observe that the value of deadline $T_{\mathrm{D}}$ that minimizes the average AoI becomes larger as the value of $\lambda_s$ decreases. 

 \begin{figure}
	\centering
	\includegraphics[scale=0.58]{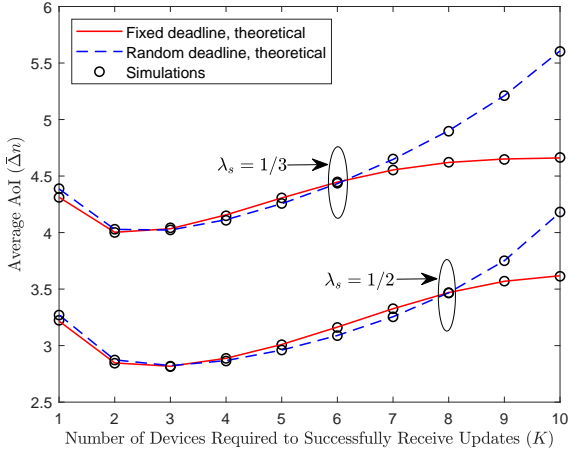}
	\caption{Average AoI versus number of IoT devices required to successfully receive each status update for different values of $\lambda_s$ when $T_{\mathrm{D}}=3$, $N = 10$, and $c=0.1$.}
	\label{fig:k}
	 \vspace{-5mm}
  \end{figure}

 Fig. \ref{fig:k} illustrates the impact of $K$ on the average AoI of the considered system for different values of $\lambda_s$ when $T_{\mathrm{D}}=3$, $N = 10$, and $c=0.1$. 
 When $K$ is small (e.g., $K = 1$), the probability that a specific device is one of the first $K$ devices that successfully receive the status update is low, and hence the average AoI is relatively large. 
 When $\lambda_s = 1/2$, by increasing the value of $K$ to 3, the probability of successful status update reception increases, which reduces the number of transmission intervals that are required to successfully receive a status update and in turn reduces the average AoI.  
 	By further increasing the value of $K$, the average length of transmission intervals increases as more devices are required to successfully receive each status update. 
	As the average length of transmission intervals increasingly dominates the AoI evolution when $K \ge 4$, the average AoI increases. 
	Therefore, with the variation of $K$, there exists a value of $K$ that balances the tradeoff between these two effects and minimizes the average AoI. 
	In addition, we observe that for smaller $K$, the average AoI for the fixed and random cases are similar, as the probability that the transmission of status updates is terminated due to the deadline is small. 
	On the other hand, when $K$ is large, the fixed deadline case outperforms the random deadline case in terms of the average AoI. 
	This is because some packets having higher deadlines in the random deadline case take a large transmission interval, leading to a larger average AoI. 
%	Similarly, we can observe that the average AoI increases as the value of $\lambda_s$ decreases. 
	
	\begin{figure}[t]
    \centering
    {\color{blue}
    \includegraphics[scale=0.58]{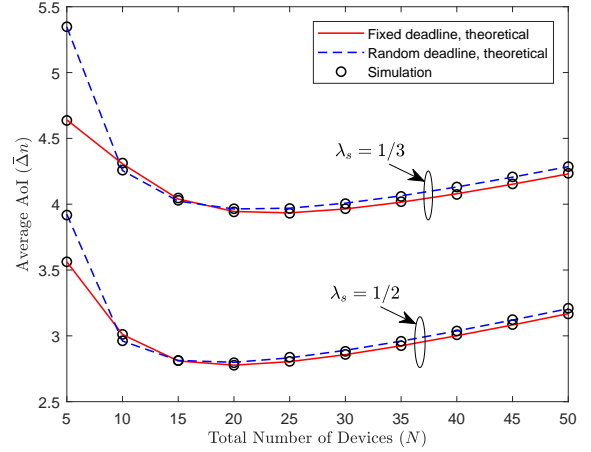}
    \caption{Average AoI versus total number of devices for different values of $\lambda_s$ when $T_{\mathrm{D}}=3$, $K = 5$, and $c=0.1$}
    \label{aoiN}}
\end{figure}

We investigate the impact of $N$ on the average AoI of the considered system for different values of $\lambda_s$ when $T_{\mathrm{D}}=3$, $K = 5$, and $c=0.1$, as shown in Fig. \ref{aoiN}. 
We can observe that the average AoI first decreases to a minimum value and then gradually increases as the value of $N$ increases. 
Specifically, when $N$ is small, the average length of transmission intervals is large, yielding a large average AoI. 
When $\lambda_s = 1/2$, by increasing the value of $N$ to 15, the average length of transmission intervals decreases, which in turn reduces the average AoI. 
By further increasing the value of $N$, the probability of a device being one of the first $K$ devices that successfully received the status update decreases, and hence, the average AoI increases. 
Similarly, we can observe that the average AoI increases as the value of $\lambda_s$ decreases. 
 
 %Meanwhile, we observe that average AoI in random case is greater than fixed case when the value of $K$ is at both ends. 
 %Specifically, when $\lambda_s = 1/2$, we can see that average AoI in random case is greater when $1\le K \le 3$ or $K\ge 8$. 
 %This is because when $K$ is small the deadline dominates the average AoI and random deadline has strong impact on average AoI  when deadline is small. 
 %When $K$ is large, the probability that IoT device receives status update successfully in random case is less than fixed case. Hence, average AoI is greater.
%

 \begin{figure}
	\centering
	\includegraphics[scale=0.58]{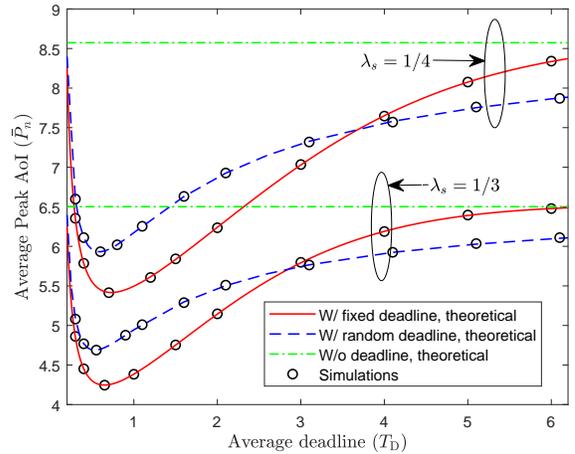}
	\caption{Average peak AoI versus average deadline $T_\mathrm{D}$ for different values of $\lambda_s$ when $K=7$, $N=10$, and $c=0.1$.}
	\label{paoid}
	 \vspace{-5mm}
  \end{figure}
  
We plot the average peak AoI versus the average deadline for different values of $\lambda_s$ when $K=7$, $N=10$, and $c=0.1$, as shown in Fig. \ref{paoid}. 
We can observe that the variation of the average peak AoI versus the average deadline has the similar trend as that observed for the average AoI in Fig. \ref{fig:s}.
In terms of the performance gap between the cases with and without deadlines, the gap for the average peak AoI is greater than that for the average AoI. 
In addition, the average peak AoI reaches the minimum point earlier than the average AoI. 
This is because the average peak AoI is more sensitive to the deadline than the average AoI.

%Motivated by that some applications may interest in the worst case age or we need to apply a threshold restricted on age. We further investigated the evolution of peak AoI under deadline constraints.  In Fig. \ref{paoid}, we show that the impact of average deadline on the peak AoI for different value of $\lambda_s$. Compare with the impact on average AoI, we can see that as average deadline increases, the evolution of average AoI and peak AoI have the same trends. The difference is that the peak AoI reaches the minimum point earlier than the average AoI, and the intersection of the random case and the fixed case also appears earlier. In particular, in fixed case when $\lambda_s=1/3$, the peak AoI reaches minimum point when $T_\mathrm{D}=0.6$ while the average AoI is $T_\mathrm{D}=0.9$. Meanwhile, when $\lambda_s=1/3$, the average AoI of two case cross at $T_\mathrm{D}=3.4$ while the peak AoI cross at $T_\mathrm{D}=2.9$. This is because  the essence of the average AoI and the peak AoI is same causes the similar trends. Meanwhile, due to the difference of the calculation ways causes the peak AoI is more sensitive to average deadline.

  \begin{figure}
	\centering
	\includegraphics[scale=0.58]{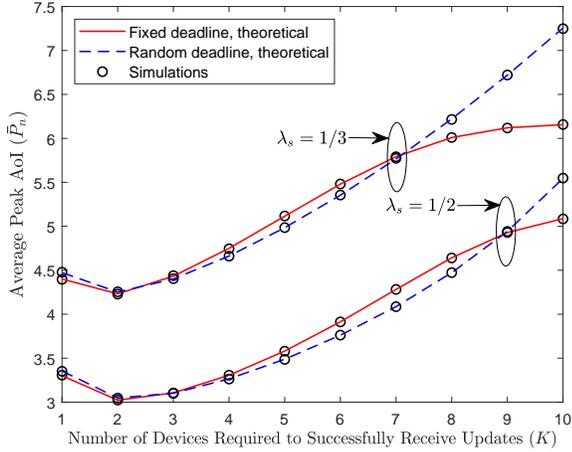}
	\caption{Average peak AoI versus number of IoT devices required to successfully receive each status update for different values of $\lambda_s$ when $T_{\mathrm{D}}=3$, $N = 10$, and $c=0.1$.}
	\label{paoik}
	 \vspace{-5mm}
  \end{figure}
 
Fig. \ref{paoik} shows the impact of $K$ on the average peak AoI for different service rates when $T_{\mathrm{D}}=3$, $N = 10$, and $c=0.1$. 
Obviously, the performance trend of the average peak AoI as $K$ increases is similar to that of the average AoI. 
The average peak AoI for fixed and random deadline cases is almost the same when $K$ is small, while the fixed deadline case achieves a lower average peak AoI than the random deadline case when $K$ is large. 
This can be explained as follows. 
With $K$ is small, only a small number of status updates are affected by the deadlines. 
On the other hand, when $K$ is large, the deadline plays an important role in determining the average transmission interval as well as the average peak AoI, and the detrimental effect of status updates with larger deadlines for the random deadline case leads to a higher average peak AoI. 

%The difference is that the area of the fixed case greater than random case becomes smaller. In particular, when $\lambda_s=1/2$, the peak AoI in fixed case is greater than random case when $3\le K \le 9$ while the average AoI is  $4\le K \le 8$ in same case. The reason is same as above mentioned, regardless the average AoI or the peak AoI there is a tradeoff between $K$ and $T_\mathrm{D}$. This means thar they are no sensitive to $K$ if $T_\mathrm{D}$ has stronger impact on them. And vice versa.

\section{Conclusions}\label{con}
In this paper, we studied the average (peak) AoI of multicast transmission with deadlines in IoT networks, where a status update is terminated by the AP if either $K$ devices successfully receive the status update or the deadline expires. 
Two categories of deadlines were considered, i.e., fixed and exponentially distributed deadlines. 
We characterized the evolution of the instantaneous AoI and derived
the distributions of the service time for all possible reception outcomes at IoT devices. 
Based on the derived distributions, we obtained the closed-form expressions of the average AoI and the average peak AoI. 
Simulations validated the theoretical analysis and showed that the deadline can be adopted to significantly reduce the average (peak) AoI. 
In particular, the deadline can be adjusted to minimize the average (peak) AoI for real-time applications. 
Results revealed that the fixed deadline achieves a lower minimum average (peak) AoI than the random deadline when optimizing the deadline. 
However, the random deadline achieves a lower average (peak) AoI than the fixed deadline in the high deadline regime. 

\section*{Appendix}
\subsection{Proof of Proposition 1}\label{pro1}
When Case $\mathcal{C}_{\mathrm{F},1}$ occurs, we have $T_{N}(K) < T_{\mathrm{D}}$ and $T_{n,j} > \min\{T_{\mathrm{D}}, T_{N}(K) \}$, which can be simplified as $T_{N}(K) < \min \{ T_{\mathrm{D}}, T_{n,j} \}$.
As a result, the CDF of the time that $K$ IoT devices successfully receive a status update conditioning on the occurrence of Case $\mathcal{C}_{\mathrm{F},1}$ can be expressed as
\begin{equation}
	\begin{split}
		& F_{T_N(K)|\mathcal{C}_{\mathrm{F},1}}(t) \\
		= & \, \mathbb{P} \left( T_N(K)<t|\mathcal{C}_{\mathrm{F},1}\right)\\
		= & \, \frac{\mathbb{P}\left( T_N(K)<t,   T_{N}(K) < \min \{ T_{\mathrm{D}}, T_{n,j} \}\right)} {\mathbb{P}( T_{N}(K) < \min \{ T_{\mathrm{D}}, T_{n,j} \})}.
	\label{ex8}
	\end{split}
\end{equation}

The numerator of \eqref{ex8} can be calculated by
\begin{eqnarray}
\hspace{-5mm}	 && \hspace{-6mm} \mathbb{P}\left( T_N(K)<t,  T_{N}(K) < \min \{ T_{\mathrm{D}}, T_{n,j} \}\right) \nonumber \\
\hspace{-5mm}	=&& \hspace{-6mm} \frac{N-K}{N}\mathbb{P}( T_N(K)<t, T_N(K)<T_{\mathrm{D}}) \nonumber \\
\hspace{-5mm}	= && \hspace{-6mm} \frac{N-K}{N} \left(1- \mathcal{Z}_{K}- \sum_{j=0}^{K-1}B_{K,j}\frac{ e^{-\lambda_s(t-c)U_{K,j}}-V_{K,j}}{{U_{K,j}}}\right),
	   \label{ex9}
\end{eqnarray}
% \begin{eqnarray}
% \mathbb{P}\left(T_N(K)<t, T_{N}(K) < \min \{ T_{\mathrm{D}}, T_{n,j} \}\right) \nonumber\\
% && \hspace{-6mm} =\frac{N-K}{N}\mathbb{P}(T_N(K)<t,T_N(K)<T_{\mathrm{D}}) \nonumber \\
%   && \hspace{-6mm} =\frac{N-K}{N}\left(1-\mathcal{Z}_{K,j}\right)- K\binom{N-1}{K} \sum_{j=0}^{K-1}\binom{K-1}{j} (-1)^j \frac{ \mathcal{V}(t)}{\mathcal{R}},
%   \label{ex9}
% \end{eqnarray}
where $V_{K,j}$, $B_{K,j}$, $U_{K,j}$ and $\mathcal{Z}_{K}$ are defined in Proposition 1.
% , $\mathbb{P}\left(T_{n,j}> T_N(K) \right)=\frac{N-K}{N}$ as $\{ T_{n,j}, n \in \mathcal{N} \}$ are identically distributed.

On the other hand, the denominator of \eqref{ex8} is given by
\begin{equation}
 \begin{split}
	 \mathbb{P}\left( T_{N}(K) < \min \{ T_{\mathrm{D}}, T_{n,j} \} \right)&=\frac{N-K}{N} \mathbb{P}\left(T_N(K)\le T_{\mathrm{D}} \right) \\
	 &= \frac{N-K}{N}\left(1- \mathcal{Z}_{K}\right).
	 \label{ex10}
 \end{split}
\end{equation}
% \begin{eqnarray}
%  \mathbb{P}\left(T_{N}(K) < \min \{ T_{\mathrm{D}}, T_{n,j} \} \right)&& \hspace{-6mm}=\frac{N-K}{N} \mathbb{P}\left(T_N(K)\le T_{\mathrm{D}} \right) = \frac{N-K}{N}\left(1-\mathcal{Z}_{K,j}\right).
%  \label{ex10}
% \end{eqnarray}

By substituting (\ref{ex9}) and (\ref{ex10}) into (\ref{ex8}), we obtain the conditional CDF of $T_N(K)$, i.e., $F_{T_N(K)|\mathcal{C}_{\mathrm{F},1}}(t)$, as follows
\begin{equation}
 \begin{split}
\hspace{-3mm} F_{T_N(K)|\mathcal{C}_{\mathrm{F},1}}(t) \!= 1 -  \sum_{j=0}^{K-1} B_{K,j}\frac{ e^{-\lambda_s(t-c)U_{K,j}}-V_{K,j}}{U_{K,j}(1-\mathcal{Z}_{K})}.
 \end{split}
 \label{ex100}
\end{equation}

As a result, the conditional first and second moments of $T_N(K)$,
i.e., $\mathbb{E} \left[ T_N(K) \left| \mathcal{C}_{\mathrm{F},1} \right. \right]$ and $\mathbb{E} \left[ T_N(K)^2 \left| \mathcal{C}_{\mathrm{F},1} \right. \right]$,
can be written as
\begin{equation}
	\begin{split}
 \mathbb{E}\left[T_N(K) \left| \mathcal{C}_{\mathrm{F},1} \right. \right]  &=\int_{c}^{T_{\mathrm{D}}} t f_{T_N(K)|\mathcal{C}_{\mathrm{F},1}}(t) \mathrm{d}t \\
 &= \frac{1}{1- \mathcal{Z}_{K}}\sum_{j=0}^{K-1} \frac{B_{K,j} }{\lambda_s U_{K,j}^2}\Big[1+c\lambda_s U_{K,j}  \\
 & \hspace{3mm} - (1 + T_{\mathrm{D}}\lambda_s U_{K,j}  )V_{K,j}\Big], 
 	\label{ex11}
	\end{split}
\end{equation}
\begin{equation}
 \begin{split}
   \mathbb{E}\left[T_N^2(K) \left| \mathcal{C}_{\mathrm{F},1} \right. \right]   & = \int_{c}^{T_{\mathrm{D}}} t^2 f_{T_N(K)|\mathcal{C}_{\mathrm{F},1}}(t) \mathrm{d}t \\
   &=\frac{1}{1- \mathcal{Z}_{K}} \sum_{j=0}^{K-1} \frac{B_{K,j}}{\lambda_s^2 U_{K,j}^3} \Big[ (1 + c \lambda_s U_{K,j})^2 \\
   & \hspace{3mm} - \left((1+T_{\mathrm{D}}\lambda_s U_{K,j}  )^2 +1\right)V_{K,j} \Big],
  \end{split}
\label{ex51}
\end{equation}
where $\mathcal{Z}_{K}$, $B_{K,j}$, $U_{K,j}$ and $V_{K,j}$ are given in Proposition 1, and $f_{T_N(K)|\mathcal{C}_{\mathrm{F},1}}(t)$ is the first derivative of $F_{T_N(K)|\mathcal{C}_{\mathrm{F},1}}(t)$ and denotes the conditional PDF of $T_N(K)$.

\subsection{Proof of Proposition 2}
The probability that Case $\mathcal{C}_{\mathrm{F},2}$ occurs can be expressed as 
\begin{eqnarray} \label{Eq_CF2}
\mathbb{P}(\mathcal{C}_{\mathrm{F},2}) && \hspace{-6mm} = \frac{ \mathbb{P}\left(T_{\mathrm{D}} < \min \{ T_N(K), T_{n,j} \} \right) } { \mathbb{P}\left(T_{n,j}>\min \{ T_{\mathrm{D}}, T_N(K) \} \right) },
\label{ex12}
\end{eqnarray}
where the denominator $\mathbb{P}\left(T_{n,j}>\min \{ T_{\mathrm{D}}, T_N(K) \} \right) = \mathbb{P}\left(T_{n,j}>T_{\mathrm{D}} \right)+\mathbb{P} \left(T_{n,j}> T_N(K) \right)-\mathbb{P}\left(T_{n,j}>T_{\mathrm{D}},T_{n,j}> T_N(K) \right)$. 
By definition, we have $\mathbb{P}\left(T_{n,j}>T_{\mathrm{D}} \right) = \mathrm{e}^{-\lambda_s (T_{\mathrm{D}}-c)}$ and $\mathbb{P}\left(T_{n,j}> T_N(K) \right)=\frac{N-K}{N}$. 
In addition, the probability that the service time of device $n$ is greater than both the deadline and the $K$-th order statistics of service times is given by 
\begin{equation}
 \begin{split}
	 &\mathbb{P}\left(T_{n,j}>T_{\mathrm{D}},T_{n,j}>T_N(K) \right) \\
	 = &\sum_{h=K+1}^{N} \mathbb{P} \left(T_{n,j}>T_{\mathrm{D}},T_{n,j}=T_N(h) \right) \\
	 = & \frac{1}{N}\sum_{h=K+1}^{N}\mathcal{Z}_{h},
	\label{ex13}
 \end{split}
\end{equation}
% \begin{eqnarray}
%  \mathbb{P}\left(T_{n,j}>T_{\mathrm{D}},T_{n,j}>T_N(K) \right) && \hspace{-6mm} = \sum_{h=K+1}^{N} \mathbb{P} \left(T_{n,j}>T_{\mathrm{D}},T_{n,j}=T_N(h) \right) =\frac{1}{N}\sum_{h=K+1}^{N}\mathcal{Z}_{h,j},
% \label{ex13}
% \end{eqnarray}
where $\mathcal{Z}_{h}$ is defined in Proposition 1. 
On the other hand, the numerator of (\ref{ex12}) can be calculated by
\begin{equation}
 \begin{split}
	 &\mathbb{P}\left(T_{\mathrm{D}} <  \min \{T_N(K),T_{n,j} \} \right) \\
	 = \,&\mathbb{P}\left(T_{\mathrm{D}} < T_N(K) < T_{n,j} \right)+\mathbb{P} \left(T_{\mathrm{D}}<T_{n,j} \le T_N(K) \right)  \\
	 =& \frac{N-K}{N}\mathcal{Z}_{K}+\sum_{h=1}^{K}  \mathcal{Z}_{h}.
	 \label{ex15}
 \end{split}
\end{equation}
% \begin{eqnarray}
% \mathbb{P}\left(T_{\mathrm{D}} <  \min \{T_N(K),T_{n,j} \} \right) && \hspace{-6mm} =\mathbb{P}\left(T_{\mathrm{D}} < T_N(K), T_N(K) < T_{n,j} \right)+\mathbb{P} \left(T_{\mathrm{D}}<T_{n,j}, T_{n,j}\le T_N(K) \right) \nonumber \\
% && \hspace{-6 mm} = \frac{N-K}{N}\mathcal{Z}_{K,j}+\sum_{h=1}^{K} \mathcal{Z}_{h,j},
% \label{ex15}
% \end{eqnarray}
% where $\mathbb{P}\left(T_{\mathrm{D}}<T_N(h) \right)$ is given in (\ref{ex14}).
% Hence, by combining (\ref{ex13}), (\ref{ex14}), and (\ref{ex15}), we obtain the conditional probability $\mathbb{P} \left( \mathcal{C}_{\mathrm{F},2} | \mathcal{C}_{\mathrm{F}} \right)$ in (\ref{ex12}).

By substituting \eqref{ex13} and \eqref{ex15} into \eqref{ex12}, we obtain $\mathbb{P} \left( \mathcal{C}_{\mathrm{F},2}\right)$.

\subsection{Proof of Proposition 3}\label{pro3}
When Case $\mathcal{C}_{\mathrm{S},1}$ occurs, we have $T_N(K)<T_\mathrm{D}$ and $T_{n,j}\le \min \{ T_{\mathrm{D}}, T_N(K) \}$, which can be simplified as $T_{n,j}<T_N(K) \le T_{\mathrm{D}} $.
As a result, the CDF of the time that $K$ devices successfully receive a status update conditioning on the occurrence of Case $\mathcal{C}_{\mathrm{S},1}$ can be expressed as
\begin{equation}
 \begin{split}
 F_{  T_N(K)| \mathcal{C}_{\mathrm{S},1} }(t) &= \mathbb{P} \left(  T_N(K)<t|  T_{n,j}<T_N(K) \le T_{\mathrm{D}} \right) \\
  &= \frac{\mathbb{P}( T_N(K)<t, T_{n,j}<T_N(K) \le T_{\mathrm{D}})}{ \mathbb{P}\left(  T_{n,j}<T_N(K) \le T_{\mathrm{D}}\right) }.
 \end{split}
\label{ex25rr}
\end{equation}

The numerator of \eqref{ex25rr} can be calculate as
\begin{equation}
 \begin{split}
  &\mathbb{P}( T_N(K)<t,  T_{n,j}<T_N(K) \le T_{\mathrm{D}}) \\
	=&\frac{K}{N}\left(\mathbb{P}( T_N(K)<T_{\mathrm{D}})-\mathbb{P}(T_N(K)\ge t, T_N(K)<T_{\mathrm{D}}) \right) \\
  =&\frac{K}{N}\left(1-\mathcal{Z}_K\right)- \frac{K}{N}\sum_{j=0}^{K-1}B_{K,j} \frac{e^{-\lambda_s U_{K,j}(t-c)}-V_{K,j}}{U_{K,j}}.
 \end{split}
 \label{ex26rr}
\end{equation}

On the other hand, the denominator of \eqref{ex25rr} is
\begin{equation}
 \begin{split}
 \mathbb{P}(T_{n,j}<T_N(K) \le T_{\mathrm{D}})=\frac{K}{N}\left(1-\mathcal{Z}_K\right).
 \end{split}
 \label{ex27rr}
\end{equation}

 By substituting \eqref{ex26rr} and \eqref{ex27rr} into \eqref{ex25rr}, we obtain the conditional CDF of $T_N(K)$, i.e., $F_{T_N(K)|\mathcal{C}_{\mathrm{S},1}}(t)$, as follows
 \begin{equation}
  \begin{split}
 \hspace{-3mm} F_{T_N(K)|\mathcal{C}_{\mathrm{S},1}}(t)=1- \sum_{j=0}^{K-1}B_{K,j} \frac{e^{-\lambda_s U_{K,j}(t-c)}-V_{K,j}}{U_{K,j}(1-\mathcal{Z}_K)}.
  \end{split}
  \label{ex101rr}
 \end{equation}

By observing that \eqref{ex101rr} equals \eqref{ex100}, we have
 \begin{equation}
  \begin{split}
  \mathbb{E}[T_N(K)|\mathcal{C}_{\mathrm{S},1}] & = \mathbb{E}\left[T_N(K) \left| \mathcal{C}_{\mathrm{F},1} \right. \right], \\
  \mathbb{E}[T_N^2(K)|\mathcal{C}_{\mathrm{S},1}] & =\mathbb{E}\left[T_N^2(K) \left| \mathcal{C}_{\mathrm{F},1} \right. \right].
  \end{split}
  \label{ex102rr}
 \end{equation}

\subsection{Proof of Proposition 4}\label{pro4}
When Case $\mathcal{C}_{\mathrm{F},1}$ occurs, i.e., $T_{N}(K) < \min \{ T_{\mathrm{D,j}}, T_{n,j} \}$, the CDF of $T_N(K)$ is given by 
\begin{equation}
	\begin{split}
		&F_{T_N(K)|\mathcal{C}_{\mathrm{F},1}}(t) \\
		= &\, \mathbb{P}(T_N(K)<t |\mathcal{C}_{\mathrm{F},1} ) \\
		=&\, \frac{\mathbb{P}( T_N(K)<t, T_N(K) < \min\{T_{\mathrm{D}, j},T_{n,j}\})} {\mathbb{P}( T_N(K) < \min\{T_{\mathrm{D}, j},T_{n,j} \}, T_N(K)\le T_{\mathrm{D}, j})}.
		\label{eq46r}
	\end{split}
\end{equation}
The numerator of \eqref{eq46r} can be calculated by
\begin{equation}
	\begin{split}
		&\mathbb{P}( T_N(K)<t, T_N(K) < \min\{T_{\mathrm{D}, j},T_{n,j}\})  \\
		=& \, \frac{N-K}{N}\mathbb{P}( T_N(K)<t, T_N(K)\le T_{\mathrm{D}, j}) \\
		=& \, \frac{N-K}{N}\big[\mathbb{P}( T_N(K)\le T_{\mathrm{D}, j}) \\
		&- \, \mathbb{P}(T_N(K)\ge t, T_N(K)\le T_{\mathrm{D}, j})\big],
		\label{eq47r}
	\end{split}
\end{equation}
where $\mathbb{P}( T_N(K)\le T_{\mathrm{D}, j})=1-\mathcal{R}_K$ and $\mathbb{P}(T_N(K)\ge t, T_N(K)\le T_{\mathrm{D}, j})$ is given by 
\begin{equation}
	\begin{split}
		&\mathbb{P}\left(T_N(K)\ge t, T_N(K)\le T_{\mathrm{D}, j})\right) \\
		=& \int_{t}^{+\infty}f_{T_{\mathrm{D}, j}}(x) \int_{x}^{+\infty}f_{T_{N}(K)}(y)\,\mathrm{d}y\,\mathrm{d}x\\
		=& \sum_{j=0}^{K-1}B_{K,j}\lambda_s\frac{e^{-H_{K,j}(t-c)}}{H_{K,j}}. 
		\label{eq48r}
	\end{split}
\end{equation}

Besides, the denominator of \eqref{eq46r} can be calculated as
\begin{equation}
	\begin{split}
		&\mathbb{P}( T_N(K) < \min\{T_{\mathrm{D}, j},T_{n,j} \}, T_N(K)\le T_{\mathrm{D}, j}) \\
		=&\sum_{h=K+1}^{N}\mathbb{P}\left(T_{\mathrm{D}, j}\ge T_{N}(K) , T_{n,j}=T_{N}(h) \right)\\
		=&\frac{N-K}{N}(1-\mathcal{R}_K). 
		\label{eq49r}
	\end{split}
\end{equation}
By substituting \eqref{eq47r}, \eqref{eq48r}, and \eqref{eq49r} into \eqref{eq46r}, we obtain the conditional CDF of $T_N(K)$, the first derivative of which is given by 
\begin{equation}
	\begin{split}
		f_{T_N(K)|\mathcal{C}_{\mathrm{F},1}}(t)=\frac{ \sum_{j=0}^{K-1}B_{K,j}N\lambda_se^{-H_{K,j}(t-c)} } {(N-K)\mathcal{R}_K}. 
		\label{eq50rb}
	\end{split}
\end{equation}

As a result, the corresponding conditional first and second moments of $T_N(K)$ can, respectively, be expressed as 
\begin{eqnarray}
	&\mathbb{E}\left[T_N(K) \left| \mathcal{C}_{\mathrm{F},1} \right. \right] =\sum_{j=0}^{K-1}B_{K,j}\lambda_s\frac{cNH_{K,j}+1}{H_{K,j}^2(N-K)(1-\mathcal{R}_K)}, \\
	&\mathbb{E}\left[T_N^2(K) \left| \mathcal{C}_{\mathrm{F},1} \right. \right] =\sum_{j=0}^{K-1}B_{K,j}\lambda_s\frac{c^2 N H_{K,j}^2 +2cH_{K,j}+2}{H_{K,j}^3(N-K)(1-\mathcal{R}_K)}.
\end{eqnarray}

\subsection{Proof of Proposition 5}\label{pro5}
When Case $\mathcal{C}_{\mathrm{F},2}$ occurs, we have  $T_{\mathrm{D}, j}< T_N(K)$ and $T_{n,j} > \min\{T_{\mathrm{D}, j}, T_{N}(K) \}$, which can be simplified as $T_{\mathrm{D}, j}< \min\{T_{n,j}, T_N(K)\}$. 
As a result, the CDF of $T_{\mathrm{D}, j}$ conditioning on the occurrence of Case $\mathcal{C}_{\mathrm{F},2}$ is given by
\begin{equation}
	\begin{split}
	&F_{T_{\mathrm{D}, j} \vert T_{\mathrm{D}, j}< \min\{T_{n,j}, T_N(K)\}}(t)\\
%		=&\mathbb{P}( T_{\mathrm{D}, j}<t | T_{\mathrm{D}, j}< \min\{T_{n,j}, T_N(K)\} ) \\
		=&\frac{\mathbb{P}( T_{\mathrm{D}, j}<t, T_{\mathrm{D}, j}<T_N(K), T_{\mathrm{D}, j}<T_{n,j})}{\mathbb{P}( T_{\mathrm{D}, j}<T_N(K),  T_{\mathrm{D}, j}<T_{n,j})}.
		\label{eq36r}
	\end{split}
\end{equation}

We calculate the numerator in \eqref{eq36r} as 
\begin{equation}
	\begin{split}
			&\mathbb{P}( T_{\mathrm{D}, j}<t,  T_{\mathrm{D}, j}<T_N(K),  T_{\mathrm{D}, j}<T_{n,j}) \\
			=&\,\mathbb{P}( T_{\mathrm{D}, j}<t,  T_{\mathrm{D}, j}<T_{n,j},  T_{n,j}\le T_N(K)) \\
			& + \mathbb{P}( T_{\mathrm{D}, j}<t,  T_{\mathrm{D}, j}<T_N(K),  T_N(K)<T_{n,j}).
		\label{eq37r}
	\end{split}
\end{equation}
The first term on the right hand side of \eqref{eq37r} can be calculated by
\begin{equation}
	\begin{split}
			&\mathbb{P}( T_{\mathrm{D}, j}<t,  T_{\mathrm{D}, j}<T_{n,j},  T_{n,j}\le T_N(K))\\ =&\sum_{h=1}^{K}\mathbb{P}( T_{\mathrm{D}, j}<t,  T_{\mathrm{D}, j}<T_N(h), T_{n,j}=T_N(h)) \\
			=&\frac{1}{N}\sum_{h=1}^{K} \left[\mathcal{R}_h-\mathbb{P}\left(T_{\mathrm{D}, j}\ge t,   T_{\mathrm{D}, j}<T_N(h)\right)\right],
		\label{eq38r}
	\end{split}
\end{equation}
where $\mathbb{P}\left(T_{\mathrm{D}, j}\ge t,   T_{\mathrm{D}, j}<T_N(h)\right)$  is given by
\begin{equation}
	\begin{split}
			&\mathbb{P}(T_{\mathrm{D}, j}\ge t,  T_{\mathrm{D}, j}<T_N(h))\\
			=& \int_{t}^{+\infty}f_{T_{\mathrm{D}, j}}(x)\int_{x}^{+\infty}f_{T_N(h)}(y) \,dy \,dx\\
			=& \sum_{j=0}^{h-1}B_{h,j}\lambda_d\frac{e^{-H_{h,j}(t-c)}}{U_{h,j}H_{h,j}}. 
		\label{eq40r}
	\end{split}
\end{equation}
The second term on the right hand side of \eqref{eq37r} is given by
\begin{equation}
	\begin{split}
			&\mathbb{P}\left( T_{\mathrm{D}, j}<t, T_{\mathrm{D}, j}<T_{n,j}, T_N(K)<T_{n,j}\right) \\
			=&\sum_{K=K+1}^{N}\mathbb{P}\left( T_{\mathrm{D}, j}<t, T_{\mathrm{D}, j}<T_{n,j},T_{n,j}=T_N(K)\right) \\
			=&\frac{(N-K)\mathcal{R}_K}{N}-\sum_{j=0}^{K-1}B_{K,j}\lambda_d\frac{e^{-H_{K,j}(t-c)}}{U_{K,j}H_{K,j}}.
		\label{eq41r}
	\end{split}
\end{equation}
On the other hand, the denominator in \eqref{eq36r} is given by
\begin{equation}
	\begin{split}
			&\mathbb{P}( T_{\mathrm{D}, j}<T_N(K),  T_{\mathrm{D}, j}<T_{n,j})\\
			=& \, \mathbb{P}\left( T_{\mathrm{D}, j}<T_{n,j}, T_{n,j}\le T_N(K)\right) \\
			 &+ \mathbb{P}\left( T_{\mathrm{D}, j}<T_N(K), T_N(K)<T_{n,j}\right) \\
			=&\, \frac{1}{N}\sum_{h=1}^{K}\mathcal{R}_h + \frac{N-K}{K}\mathcal{R}_K. 
		\label{eq42r}
	\end{split}
\end{equation}

By substituting \eqref{eq38r}, \eqref{eq40r}, \eqref{eq41r}, and \eqref{eq42r} into \eqref{eq37r},  we obtain the CDF of $T_{\mathrm{D}, j}$ conditioning on the occurrence of Case $\mathcal{C}_{\mathrm{F},2}$. 
%  The PDF can be {\color{red}expressed as} 
% \begin{equation}
% 	\begin{split}
% 			&f_{T_{\mathrm{D}, j} | \mathcal{C}_{\mathrm{F},2}}(t) \\
% 			=&\left(\frac{1}{N}\sum_{h=1}^{K}h\binom{N}{h} \sum_{j=0}^{h-1}\binom{h-1}{j}(-1)^j \frac{\lambda_d}{U_{h,j}}e^{-(\lambda_sU_{h,j}+\lambda_d)t} + \right.\\
% 			& \left. K\binom{N}{K}\sum_{j=0}^{K-1}\binom{K-1}{j}(-1)^j \frac{\lambda_d}{R_j}e^{-(\lambda_sR_j+\lambda_d)t}\right) {/} \\
% 			&\left( \frac{1}{N}\sum_{h=1}^{K}\mathcal{I}_h + \frac{N-K}{K}\Psi_3\right).
% 		\label{eq43r}
% 	\end{split}
% \end{equation}

With conditional CDF $F_{T_{\mathrm{D}, j} | \mathcal{C}_{\mathrm{F},2}}(t)$, the conditional first and second moments of $T_{\mathrm{D}, j}$ can, respectively, be calculated by 
\begin{equation}
	\begin{split}
		&\mathbb{E}\left[T_{\mathrm{D}, j} \left| \mathcal{C}_{\mathrm{F},2} \right. \right] =\int_{c}^{+\infty}t  \, \mathrm{d} F_{T_{\mathrm{D}, j} | \mathcal{C}_{\mathrm{F},2}}(t) \\
		=&\frac{1}{\sum_{h=1}^{K}\mathcal{R}_h+(N-K)\mathcal{R}_K}\left(\sum_{h=1}^{K}\sum_{j=0}^{h-1}B_{h,j}\lambda_d\frac{cH_{h,j}+1}{U_{h,j}H_{h,j}^2}\right. \\
		&\left. + N\sum_{j=0}^{K-1}B_{K,j}\lambda_d\frac{cH_{K,j}+1}{U_{K,j}H_{K,j}^2}\right),
		\label{eq44r}
	\end{split}
\end{equation}
and
\begin{equation}
	\begin{split}
		&\mathbb{E}\left[T_{\mathrm{D}, j}^2 \left| \mathcal{C}_{\mathrm{F},2} \right. \right] =\int_{c}^{+\infty}t^2 \,\mathrm{d} F_{T_{\mathrm{D}, j} | \mathcal{C}_{\mathrm{F},2}}(t) \\
		=&\frac{1}{\sum_{h=1}^{K}\mathcal{R}_h+(N-K)\mathcal{R}_K}\left(\sum_{h=1}^{K}\sum_{j=0}^{h-1}B_{h,j}\lambda_d\frac{c^2H_{h,j}^2+2cH_{h,j}+2}{U_{h,j}H_{h,j}^3}\right.\\
		&\left.+N\sum_{j=0}^{K-1}B_{K,j}\lambda_d\frac{c^2H_{K,j}^2+2cH_{K,j}+2}{U_{K,j}H_{K,j}^3}\right). 
		\label{eq45r}
	\end{split}
\end{equation}

\subsection{Proof of Proposition 6}\label{pro6}
When Case $\mathcal{C}_{\mathrm{S},2}$ occurs, i.e.,  $T_{n,j}\le T_{\mathrm{D}, j}< T_N(K)$, the CDF of $T_{\mathrm{D}, j}$ conditioning on the occurrence of Case $\mathcal{C}_{\mathrm{S},2}$ can be expressed as 
\begin{equation}
	\begin{split}
		&F_{T_{\mathrm{D}, j} | T_{n,j}\le T_{\mathrm{D}, j}< T_N(K)}(t)\\
		=&\frac{\mathbb{P}( T_{\mathrm{D}, j}<t, T_{n,j}\le T_{\mathrm{D}, j}< T_N(K))} {\mathbb{P}( T_{n,j}\le T_{\mathrm{D}, j}< T_N(K))}.
		\label{eq57r}
	\end{split}
\end{equation}

The numerator of \eqref{eq57r} can be calculated by
\begin{equation}
	\begin{split}
		&\; \mathbb{P}( T_{\mathrm{D}, j}<t, T_{n,j}\le T_{\mathrm{D}, j}< T_N(K))\\
		= & \;\mathbb{P}\left( T_{\mathrm{D}, j}<t, T_{\mathrm{D}, j}\le T_N(K), T_{n,j}\le T_N(K)\right)\\
		 &  - \mathbb{P}\left( T_{\mathrm{D}, j}<t, T_{\mathrm{D}, j}<T_{n,j}, T_{n,j}\le T_N(K)\right).
		\label{eq57rr}
	\end{split}
\end{equation}

The first term on the right hand side of \eqref{eq57rr} can be calculated as
\begin{equation}
	\begin{split}
		&\mathbb{P}\left( T_{\mathrm{D}, j}<t, T_{\mathrm{D}, j}\le T_N(K), T_{n,j}\le T_N(K)\right)\\
		=&\frac{K}{N}\sum_{j=0}^{k-1}B_{K,j}\lambda_d\frac{1-e^{-H_{K,j}(t-c)}}{H_{K,j}}.
		\label{eq58r}
	\end{split}
\end{equation}

On the other hand, the second term on the right hand side of \eqref{eq57rr} is given by
\begin{equation}
	\begin{split}
		& \mathbb{P}( T_{\mathrm{D}, j}<t, T_{\mathrm{D}, j}<T_{n,j}, T_{n,j}\le T_N(K))\\
		= & \frac{1}{N}\sum_{h=1}^{K}\sum_{j=0}^{h-1}B_{h,j}\lambda_d\frac{1-e^{-H_{h,j}(t-c)}}{H_{h,j}}.
		\label{eq59r}
	\end{split}
\end{equation}

Then, the denominator in \eqref{eq57r} can be calculated by
\begin{equation}
	\begin{split}
		\mathbb{P}\left( T_{n,j}\le T_{\mathrm{D}, j}< T_N(K)\right)=\frac{K}{N} \mathcal{R}_K- \frac{1}{N} \sum_{h=1}^{K} \mathcal{R}_h.
		\label{eq59rr}
	\end{split}
\end{equation}

By substituting \eqref{eq57rr}, \eqref{eq58r}, \eqref{eq59r}, and \eqref{eq59rr} into \eqref{eq57r}, we obtain the conditional CDF of $T_{\mathrm{D}, j}$. 
As a result, the conditional first and second moments of $T_{\mathrm{D}, j}$ are given by
\begin{equation}
	\begin{split}
		&\mathbb{E}\left[T_{\mathrm{D}, j} \left| \mathcal{C}_{\mathrm{S},2} \right. \right] =\int_{c}^{+\infty}t \,\mathrm{d} F_{T_{\mathrm{D}, j} | T_{n,j}\le T_{\mathrm{D}, j}< T_N(K)}(t)\\
		=&\frac{1}{K\mathcal{R}_K - \sum_{h=1}^{K}\mathcal{R}_h}\left(K\sum_{j=0}^{K-1}B_{K,j}\lambda_d\frac{cH_{K,j}+1}{U_{K,j}H_{K,j}^2}\right. \\
		&\left. - \sum_{h=1}^{K}\sum_{j=0}^{h-1}B_{h,j}\lambda_d\frac{cH_{h,j}+1}{U_{h,j}H_{h,j}^2} \right),
	\end{split}
\end{equation}
and
\begin{equation}
	\begin{split}
		&\mathbb{E}\left[T_{\mathrm{D}, j}^2 \left| \mathcal{C}_{\mathrm{S},2} \right. \right] =\int_{c}^{+\infty}t^2 \,\mathrm{d} F_{T_{\mathrm{D}, j} | T_{n,j}\le T_{\mathrm{D}, j}< T_N(K)}(t)\\
		=&\frac{1}{K\mathcal{R}_K - \sum_{h=1}^{K}\mathcal{R}_h}\left(K\sum_{j=0}^{K-1}B_{K,j}\lambda_d\frac{c^2H_{K,j}^2+2cH_{K,j}+2}{U_{K,j}H_{K,j}^3}\right. \\
		&\left. - \sum_{h=1}^{K}\sum_{j=0}^{h-1}B_{h,j}\lambda_d\frac{c^2H_{h,j}^2+2cH_{h,j}+2}{U_{h,j}H_{h,j}^3} \right).
		\label{eq62r}
	\end{split}
\end{equation}

% \subsection{Proof of Proposition 6}
%  Same as Appendix B.
% The CDF of $T_N(K)$ conditioning on $T_{n,j}<T_N(K)<T_{\mathrm{D}, j}$ is given by
% \begin{equation}
% 	\begin{split}
% 		&F_{T_N(K) | T_{n,j}<T_N(K)<T_{\mathrm{D}, j}}(t)\\
% 		&=P( T_N(K)<t |   T_N(K)\le T_{\mathrm{D}, j},T_{n,j}\le T_N(K)) \\
% 		&= 1- \frac{1}
% 		{1-\mathcal{R}_K}\sum_{j=0}^{K-1}B_{K,j}\lambda_s\frac{ e^{H_{K,j}(t-c)}}{H_{K,j}}
% 		\label{eq53r}
% 	\end{split}
% \end{equation}

% Therefore, the first and second moments of conditional $T_N(K)$ are respectively given by
% \begin{equation}
% 	\begin{split}
% 		\mathbb{E}\left[T_N(K) \left| \mathcal{C}_{\mathrm{S},1} \right. \right]=\mathbb{E}\left[T_N(K) \left| \mathcal{C}_{\mathrm{F},1} \right. \right],
% 		\label{eq55r}
% 	\end{split}
% \end{equation}
% \begin{equation}
% 	\begin{split}
% 		\mathbb{E}\left[T_N^2(K) \left| \mathcal{C}_{\mathrm{S},1} \right. \right]
% 		=\mathbb{E}\left[T_N^2(K) \left| \mathcal{C}_{\mathrm{F},1} \right. \right],
% 		\label{eq56r}
% 	\end{split}
% \end{equation}
% where $\mathbb{E}\left[T_N(K) \left| \mathcal{C}_{\mathrm{F},1} \right. \right]$ and $\mathbb{E}\left[T_N^2(K) \left| \mathcal{C}_{\mathrm{F},1} \right. \right]$ are given in Proposition 3.

\bibliographystyle{IEEEtran}

\bibliography{refs}

\begin{thebibliography}{10}
\providecommand{\url}[1]{#1}
\csname url@samestyle\endcsname
\providecommand{\newblock}{\relax}
\providecommand{\bibinfo}[2]{#2}
\providecommand{\BIBentrySTDinterwordspacing}{\spaceskip=0pt\relax}
\providecommand{\BIBentryALTinterwordstretchfactor}{4}
\providecommand{\BIBentryALTinterwordspacing}{\spaceskip=\fontdimen2\font plus
\BIBentryALTinterwordstretchfactor\fontdimen3\font minus
  \fontdimen4\font\relax}
\providecommand{\BIBforeignlanguage}[2]{{%
\expandafter\ifx\csname l@#1\endcsname\relax
\typeout{** WARNING: IEEEtran.bst: No hyphenation pattern has been}%
\typeout{** loaded for the language `#1'. Using the pattern for}%
\typeout{** the default language instead.}%
\else
\language=\csname l@#1\endcsname
\fi
#2}}
\providecommand{\BIBdecl}{\relax}
\BIBdecl

\bibitem{gubbi2013internet}
J.~Gubbi, R.~Buyya, S.~Marusic, and M.~Palaniswami, ``{Internet of Things
  (IoT): A vision, architectural elements, and future directions},''
  \emph{Future generation computer systems}, vol.~29, no.~7, pp. 1645--1660,
  Jul. 2013.

\bibitem{kosta2017age}
A.~Kosta, N.~Pappas, V.~Angelakis \emph{et~al.}, ``{Age of information: A new
  concept, metric, and tool},'' \emph{Foundations and Trends{\textregistered}
  in Networking}, vol.~12, no.~3, pp. 162--259, 2017.

\bibitem{kim2014scheduling}
K.~S. Kim, C.-p. Li, and E.~Modiano, ``Scheduling multicast traffic with
  deadlines in wireless networks,'' in \emph{Proc. IEEE INFOCOM}, Toronto,
  Canada, Apr. 2014.

\bibitem{lyu2019characterizing}
F.~Lyu, H.~Zhu, N.~Cheng, H.~Zhou, W.~Xu, M.~Li, and X.~Shen, ``Characterizing
  urban vehicle-to-vehicle communications for reliable safety applications,''
  \emph{IEEE Trans. Intell. Transp. Syst.}, 2020.

\bibitem{mao2014optimal}
Z.~Mao, C.~E. Koksal, and N.~B. Shroff, ``Optimal online scheduling with
  arbitrary hard deadlines in multihop communication networks,'' \emph{IEEE/ACM
  Trans. Netw.}, vol.~24, no.~1, pp. 177--189, Jan. 2014.

\bibitem{akyurek2018optimal}
A.~S. Akyurek and T.~S. Rosing, ``Optimal packet aggregation scheduling in
  wireless networks,'' \emph{IEEE Trans. Mobile Comput.}, vol.~17, no.~12, pp.
  2835--2852, Dec. 2018.

\bibitem{on}
C.~{Kam}, S.~{Kompella}, G.~D. {Nguyen}, J.~E. {Wieselthier}, and
  A.~{Ephremides}, ``On the age of information with packet deadlines,''
  \emph{IEEE Trans. Inf. Theory}, vol.~64, no.~9, pp. 6419--6428, Sep. 2018.

\bibitem{sun2017update}
Y.~Sun, E.~Uysal-Biyikoglu, R.~D. Yates, C.~E. Koksal, and N.~B. Shroff,
  ``Update or wait: How to keep your data fresh,'' \emph{IEEE Trans. Inf.
  Theory}, vol.~63, no.~11, pp. 7492--7508, Nov. 2017.

\bibitem{statusqueue}
S.~K. {Kaul}, R.~D. {Yates}, and M.~{Gruteser}, ``Status updates through
  queues,'' in \emph{Proc. IEEE CISS}, Princeton, NJ, Mar. 2012.

\bibitem{real}
S.~{Kaul}, R.~{Yates}, and M.~{Gruteser}, ``Real-time status: How often should
  one update?'' in \emph{Proc. IEEE INFOCOM}, Orlando, FL, Mar. 2012.

\bibitem{ontheage}
M.~{Costa}, M.~{Codreanu}, and A.~{Ephremides}, ``On the age of information in
  status update systems with packet management,'' \emph{IEEE Trans. Inf.
  Theory}, vol.~62, no.~4, pp. 1897--1910, Apr. 2016.

\bibitem{yates2012real}
R.~D. Yates and S.~Kaul, ``Real-time status updating: Multiple sources,'' in
  \emph{Proc. IEEE ISIT}, Boston, MA, Jul. 2012.

\bibitem{optimizing}
L.~{Huang} and E.~{Modiano}, ``Optimizing age-of-information in a multi-class
  queueing system,'' in \emph{Proc. IEEE ISIT}, Hong Kong, Jun. 2015.

\bibitem{kuang2019age}
Q.~Kuang, J.~Gong, X.~Chen, and X.~Ma, ``Age-of-information for
  computation-intensive messages in mobile edge computing,'' in \emph{Proc.
  IEEE WCSP}, Xi an, China, Oct. 2019.

\bibitem{boundson}
S.~Farazi, A.~G. Klein, and D.~R. Brown, ``Bounds on the age of information for
  global channel state dissemination in fully-connected networks,'' in
  \emph{Proc. IEEE ICCCN}, Vancouver, Canada, Aug. 2017.

\bibitem{gu2019timely}
Y.~Gu, H.~Chen, Y.~Zhou, Y.~Li, and B.~Vucetic, ``Timely status update in
  {Internet of Things} monitoring systems: An age-energy tradeoff,'' \emph{IEEE
  Internet Things J.}, vol.~6, pp. 5324--5335, Jun. 2019.

\bibitem{tandem}
C.~{Kam}, J.~P. {Molnar}, and S.~{Kompella}, ``Age of information for queues in
  tandem,'' in \emph{Proc. IEEE MILCOM}, Los Angeles, CA, Oct. 2018.

\bibitem{markov}
L.~{Huang} and L.~P. {Qian}, ``Age of information for transmissions over
  {M}arkov channels,'' in \emph{Proc. IEEE GLOBECOM}, Singapore, Dec. 2017.

\bibitem{multiplesources}
R.~D. Yates and S.~K. Kaul, ``The age of information: Real-time status updating
  by multiple sources,'' \emph{IEEE Trans. Inf. Theory}, vol.~65, no.~3, pp.
  1807--1827, 2019.

\bibitem{schedulingstatus}
B.~T. Bacinoglu and E.~Uysalbiyikoglu, ``Scheduling status updates to minimize
  age of information with an energy harvesting sensor,'' in \emph{Proc. IEEE
  ISIT}, Aachen, Germany, Jun. 2017.

\bibitem{achievingthe}
B.~T. Bacinoglu, Y.~Sun, E.~Uysalbivikoglu, and V.~Mutlu, ``Achieving the
  age-energy tradeoff with a finite-battery energy harvesting source,'' in
  \emph{Proc. IEEE ISIT}, Vail, CO, Jun. 2018.

\bibitem{arafa2018age}
A.~Arafa, J.~Yang, S.~Ulukus, and H.~V. Poor, ``Age-minimal online policies for
  energy harvesting sensors with incremental battery recharges,'' in
  \emph{Proc. IEEE ITA}, San Diego, CA, Feb. 2018.

\bibitem{elgabli2019reinforcement}
A.~Elgabli, H.~Khan, M.~Krouka, and M.~Bennis, ``Reinforcement learning based
  scheduling algorithm for optimizing age of information in ultra reliable low
  latency networks,'' in \emph{Proc. IEEE ISCC}, Barcelona, Spain, Jun. 2019.

\bibitem{multipleflows}
H.~B. {Beytur} and E.~{Uysal}, ``Age minimization of multiple flows using
  reinforcement learning,'' in \emph{Proc. IEEE ICNC}, Honolulu, HI, Feb 2019.

\bibitem{economicsof}
S.~Hao and L.~Duan, ``Economics of age of information management under network
  externalities,'' in \emph{Proc. ACM MobiHoc}, Catania, Italy, Jul. 2019.

\bibitem{agebased}
N.~Lu, B.~Ji, and B.~Li, ``Age-based scheduling: Improving data freshness for
  wireless real-time traffic,'' in \emph{Proc. ACM Mobihoc}, Los Angeles, CA,
  Jun. 2018.

\bibitem{adaptivecoding}
\BIBentryALTinterwordspacing
S.~Feng and J.~Yang, ``Adaptive coding for information freshness in a two-user
  broadcast erasure channel,'' \emph{CoRR}, vol. abs/1905.00521, 2019.
  [Online]. Available: \url{http://arxiv.org/abs/1905.00521}
\BIBentrySTDinterwordspacing

\bibitem{minimizingthe}
I.~Kadota and E.~Modiano, ``Minimizing the age of information in wireless
  networks with stochastic arrivals,'' \emph{IEEE Trans. Mobile Comput.}, 2019.

\bibitem{uav}
M.~A. Abdelmagid and H.~S. Dhillon, ``Average peak age-of-information
  minimization in uav-assisted iot networks,'' \emph{IEEE Trans. Veh.
  Technol.}, vol.~68, no.~2, pp. 2003--2008, Feb. 2019.

\bibitem{multicast}
J.~{Zhong}, R.~D. {Yates}, and E.~{Soljanin}, ``Multicast with prioritized
  delivery: How fresh is your data?'' in \emph{Proc. IEEE SPAWC}, Kalamata,
  Greece, Jun. 2018.

\bibitem{status}
J.~{Zhong}, E.~{Soljanin}, and R.~D. {Yates}, ``Status updates through
  multicast networks,'' in \emph{Proc. Allerton}, Monticello, IL, Oct. 2017,
  pp. 463--469.

\bibitem{optimum}
S.~{Nath}, J.~{Wu}, and J.~{Yang}, ``Optimum energy efficiency and
  age-of-information tradeoff in multicast scheduling,'' in \emph{Proc. IEEE
  ICC}, Kansas, MO, May 2018.

\bibitem{Buyukates2018Age}
B.~Buyukates, A.~Soysal, and S.~Ulukus, ``Age of information in two-hop
  multicast networks,'' in \emph{Proc. IEEE ACSSC}, Pacific Grove, CA, Oct.
  2018.

\bibitem{broadcastage}
M.~Wang and Y.~Dong, ``Broadcast age of information in csma/ca based wireless
  networks,'' in \emph{Proc. IEEE IWCMC}, Tangier, Morocco, Jun. 2019.

\bibitem{schedulingbroadcast}
I.~Kadota, A.~Sinha, E.~Uysal-Biyikoglu, R.~Singh, and E.~Modiano, ``Scheduling
  policies for minimizing age of information in broadcast wireless networks,''
  \emph{IEEE/ACM Trans. Netw.}, vol.~26, no.~6, pp. 2637--2650, Dec. 2018.

\bibitem{analysis}
Y.~{Inoue}, ``Analysis of the age of information with packet deadline and
  infinite buffer capacity,'' in \emph{Proc. IEEE ISIT}, Vail, CO, Jun. 2018.

\bibitem{order}
H.~A. David and H.~N. Nagaraja, ``Order statistics,'' \emph{Encyclopedia of
  Statistical Sciences}, 2004.

\end{thebibliography}

\end{document}